\documentclass[11pt]{article}
\usepackage[margin=1in]{geometry}
\usepackage[utf8]{inputenc}
\usepackage{microtype}
\usepackage[all=normal,paragraphs=tight,bibliography=tight,floats=tight]{savetrees}
\usepackage{listings}
  \usepackage{mathrsfs}

  \usepackage{amsfonts,amsthm,amsmath}
\usepackage{amstext}

\usepackage{graphicx,tikz}
\usepackage{comment}
\usepackage{url}
\usepackage{xspace}
\usepackage{todonotes}

\def\cqedsymbol{\ifmmode$\lrcorner$\else{\unskip\nobreak\hfil
\penalty50\hskip1em\null\nobreak\hfil$\lrcorner$
\parfillskip=0pt\finalhyphendemerits=0\endgraf}\fi} 

\newcommand{\cqed}{\renewcommand{\qed}{\cqedsymbol}}

\newtheorem{theorem}{Theorem}
\newtheorem{lemma}{Lemma}
\newtheorem{corollary}[lemma]{Corollary}
\newtheorem{claim}{Claim}
\theoremstyle{definition}
\newtheorem{definition}{Definition}

\newcommand{\ellmax}{\ell}
\newcommand{\Oh}{\mathcal{O}}
\newcommand{\eps}{\epsilon}
\newcommand{\crdx}[2]{x_{#1}^{(#2)}}
\newcommand{\crdy}[2]{y_{#1}^{(#2)}}
\newcommand{\height}{h}
\newcommand{\width}{g}
\newcommand{\weight}{w}
\newcommand{\Pp}{\mathcal{P}}
\newcommand{\Rr}{\mathcal{R}}
\newcommand{\Qq}{\mathcal{Q}}
\newcommand{\Ss}{\mathcal{S}}
\newcommand{\MWISR}{{\sc{Maximum Weight Independent Set of Rectangles}}\xspace}
\newcommand{\mwisr}{{\sc{MWISR}}\xspace}

\newcommand{\OPT}{\mathsf{OPT}}
\newcommand{\Hf}{\mathsf{H}}
\newcommand{\Vf}{\mathsf{V}}
\newcommand{\Ll}{\mathcal{L}}
\newcommand{\Val}{\mathsf{Value}}
\newcommand{\glob}{{I_{\mathsf{all}}}}

\newcommand{\MWIS}{{\sc Maximum Weight Independent Set}\xspace}
\newcommand{\mwis}{{\sc MWIS}\xspace}
\newcommand{\mdi}{\ensuremath{\Delta^{\leq i}}}
\newcommand{\md}{\ensuremath{\Delta^{\leq 1}}}
\newcommand{\mdt}{\ensuremath{\Delta^{\leq 2}}}
\newcommand{\N}[1]{\ensuremath{N^{\leq 1}_{#1}}}
\newcommand{\Nt}[1]{\ensuremath{N^{\leq 2}_{#1}}}
\newcommand{\Ni}[1]{\ensuremath{N^{\leq i}_{#1}}}

\newcommand{\clmap}{\phi}

\ifdefined\DEBUG
\usepackage[normalem]{ulem}

 \newcommand{\ej}[1]{\textcolor{blue}{#1}}
 \newcommand{\andy}[1]{\textcolor{green}{#1}}
 
\else
  
  \newcommand{\ej}[1]{#1}
  \newcommand{\andy}[1]{#1}
   
\fi

\title{Approximation and parameterized algorithms for geometric independent set with shrinking\thanks{
The research of Mi.\ Pilipczuk is supported by Polish National Science Centre grant UMO-2013/11/D/ST6/03073.
Mi.\ Pilipczuk is also supported by the Foundation for Polish Science (FNP) via the START stipend programme.}}

\author{
  Micha\l{} Pilipczuk\thanks{
    Institute of Informatics, University of Warsaw, Poland, \texttt{michal.pilipczuk@mimuw.edu.pl}.
  }
  \and 
  Erik Jan van Leeuwen\thanks{
    MPI f\"ur Informatik, Saarland Informatics Campus, Saarbr\"ucken, Germany, \texttt{erikjan@mpi-inf.mpg.de}.
  }
  \and 
  Andreas Wiese\thanks{
    Departamento de Ingenieria Industrial, Universidad de Chile, Chile, \texttt{awiese@dii.uchile.cl}.
  }
}

\date{}

\begin{document}

\maketitle

\begin{abstract}
Consider the {\sc{Maximum Weight Independent Set}} problem for rectangles:
given a family of weighted axis-parallel rectangles in the plane, find a maximum-weight subset of non-overlapping rectangles.
The problem is 
notoriously hard both in the approximation and in the parameterized setting.
The best known polynomial-time approximation algorithms achieve super-constant approximation ratios~\cite{CC2009,ChanHarPeled2012}, 
even though there is a $(1+\epsilon)$-approximation running in quasi-polynomial time~\cite{AW2013,ChuzhoyEne2016}.
When parameterized by the target size of the solution, the problem is $\mathsf{W}[1]$-hard even in the unweighted setting~\cite{Marx07a}.

To achieve tractability, we study the following {\em{shrinking model}}:
one is allowed to shrink each input rectangle by a multiplicative factor $1-\delta$
for some fixed $\delta>0$, \ej{but} the performance is still compared against the optimal solution for the original, non-shrunk instance.
We prove that in this regime, the problem admits an EPTAS with running time $f(\eps,\delta)\cdot n^{\Oh(1)}$, and an FPT algorithm with running time $f(k,\delta)\cdot n^{\Oh(1)}$,
in the setting where a maximum-weight solution of size at most $k$ is to be computed.
This improves and significantly simplifies a PTAS given earlier for this problem~\cite{AdamaszekChalermsookWiese2015}, and provides the first parameterized results for the shrinking model.
Furthermore, we explore kernelization in the shrinking model, by giving efficient kernelization procedures for several variants of the \ej{problem} when the input rectangles are squares.

\end{abstract}

\section{Introduction}\label{sec:intro}

The classic {\sc{\ej{Maximum (Weight)} Independent Set}} problem is defined as follows: given a graph,
the goal is to select a set of pairwise non-adjacent vertices with maximum cardinality
or \ej{maximum} total weight. In its full generality, the problem is $\mathsf{NP}$-hard
and intractable from the perspective of both approximation and parameterized algorithms:
it is $\mathsf{NP}$-hard to approximate within ratio $n^{1-\epsilon}$ for any $\epsilon>0$~\cite{zuck07}, and it is $\mathsf{W}[1]$-hard
when parameterized by the solution size~\cite{downey1995fixed}.
Therefore, many restricted settings have been studied.

One well-studied case
is to consider a geometric setting.
That is, the input consists of a family of
geometric objects, and the goal is to select a maximum-weight \ej{subfamily} of pairwise non-overlapping objects.
This case reduces to the graph setting by 
considering
the intersection graph of the objects.
The graphs that arise in this way
are highly structured,
which gives hope for better
results than for general graphs.

This paper concentrates on the variant in which the input objects are axis-parallel rectangles in the two-dimensional plane.
In this variant, {\sc Maximum Weight Independent Set} admits much smaller approximation ratios than on general graphs.
There is an $(1+\epsilon)$-approximation algorithm
running in polynomial time (PTAS) if all input objects are squares~\cite{chan2003polynomial,EJS2005,Leeuwen06},
an $\Oh(\log\log n)$-approximation algorithm for unweighted rectangles~\cite{CC2009}, and an $\Oh(\log n/\log\log n)$-approximation
algorithm for weighted rectangles~\cite{ChanHarPeled2012}.
Recently, there has been much interest in designing {\em{quasi-polynomial approximation schemes}} ({\em{QPTASes}}), that is,
$(1+\epsilon)$-approximation algorithms running in quasi-polynomial time. This research
direction culminated in giving a QPTAS for the case of (not necessarily convex) polygons in the plane~\cite{adamaszek2014qptas,Har-Peled2014}.
Whether this can be improved to a PTAS remains a challenging open problem.

From the parameterized perspective, the problem remains $\mathsf{W}[1]$-hard when parameterized by the $k$ size of the solution, even for unweighted unit squares~\cite{Marx07a}.
Therefore, the existence of an FPT algorithm with running time $f(k)\cdot n^{\Oh(1)}$ for a computable $f$ is unlikely under standard assumptions from parameterized complexity;
this also excludes the existence of an EPTAS for the problem~\cite{Marx07a}.
However, the problem admits a faster-than-brute-force parameterized algorithm with running time $n^{\Oh(\sqrt{k})}$, which is optimal under the Exponential Time Hypothesis~\cite{MarxP15}.
This algorithm works in the general setting of finding a maximum-weight independent set of size $k$ in a family of polygons in the plane.

\paragraph*{Shrinking model.}
In order to circumvent some of the many challenges that arise when designing approximation or parameterized algorithms for geometric {\sc Maximum Weight Independent Set}, 
we investigate the {\em shrinking model} introduced by Adamaszek et al.~\cite{AdamaszekChalermsookWiese2015}. 
In this model, one is allowed to shrink each input object by a multiplicative factor $1-\delta$ for some fixed $\delta>0$,
but the weight of the computed solution is still compared to the optimum for the original, non-shrunk instance; see Section~\ref{sec:prelims} for a formal definition.
It is known that the shrinking model allows for substantially better approximation algorithms than the general setting: 
Adamaszek et al.~\cite{AdamaszekChalermsookWiese2015} gave a PTAS for axis-parallel rectangles, which was later generalized by Wiese to arbitrary convex polygons~\cite{wiese2016independent}.
However, it has not been studied so far whether shrinking also helps to design parameterized algorithms.
One concrete question would be whether {\sc Independent Set} for axis-parallel rectangles remains $\mathsf{W}[1]$-hard in the shrinking model.

\paragraph*{Our results.}
This paper addresses the parameterized complexity of {\sc{Maximum Independent Set of Rectangles}} in the shrinking model, and answers the above questions positively. 
On the way to our two main parameterized contributions, we also improve the PTAS by Adamaszek et al.~\cite{AdamaszekChalermsookWiese2015} to an EPTAS.

Our first main contribution is that {\sc{Maximum Independent Set of Rectangles}} is fixed-parameter
tractable (FPT) in the shrinking model. More precisely, for a shrinking parameter $\delta$, we can
decide in randomized time $f(k,\delta)\cdot (nN)^{\Oh(1)}$ whether there is \ej{an independent set of} 
$k$ (shrunk) rectangles, or \ej{the original family has no independent subfamily of size $k$}.
Here, $N$ is the total bit size of the input and $f$ is some computable function.
The algorithm works also in the weighted setting,
where we are looking for a maximum-weight subset of at most $k$ non-overlapping rectangles.
The main reason why we are able to circumvent the $\mathsf{W}[1]$-hardness for the standard model (i.e., without shrinking) 
is that the reduction of Marx~\cite{Marx07a} heavily relies on tiny differences in the coordinates of the rectangles. 
However, as Adamaszek et al.~\cite{AdamaszekChalermsookWiese2015} and this paper show, this aspect vanishes in the shrinking model.

The parameterized algorithm is actually a consequence of an EPTAS that we present for {\sc Maximum Weight Independent Set of Rectangles}.
That is, we give an algorithm with running time $f(\epsilon,\delta)\cdot(nN)^{\Oh(1)}$ that finds a subset of rectangles that do not overlap after shrinking by factor $1-\delta$,
and whose total weight is at least $1-\epsilon$ times the optimum without shrinking. Recall that the standard model does not admit an EPTAS, unless $\mathsf{FPT} = \mathsf{W}[1]$~\cite{Marx07a}.

Our EPTAS is based on the same principles as the PTAS of Adamaszek et al.~\cite{AdamaszekChalermsookWiese2015}.
The idea is to assemble an optimum solution using a bottom-up dynamic-programming approach pioneered by Erlebach et al.~\cite{EJS2005}. 
Each subproblem solved in the dynamic program corresponds to the maximum weight of an independent set 
contained in a ``box'', and the computation of the optimum for each such box boils down to enumerating a limited number of carefully chosen partitions of the box into smaller boxes.
Intuitively, the ability to shrink is used to make sure that rectangles fit nicely into the different boxes. 
The main challenge is to ensure that these boxes can be assumed to be simple, and therefore only a limited number of subproblems is necessary to assemble a near-optimum solution.

The crucial contribution in our approximation algorithm is that we show that rectangular boxes suffice. 
In~\cite{AdamaszekChalermsookWiese2015}, 
most rectangles were shrunk
in \emph{only one direction} and therefore, 
the
boxes were axis-parallel polygons with at most $g(\epsilon,\delta)$ sides each, for some function $g$. 
This makes the dynamic program very complex, and yields a running time of $(nN)^{g(\epsilon,\delta)}$ due to the sheer number of subproblems solved. 
In this paper, we fully exploit the properties of the shrinking model and shrink \emph{each} rectangle in \emph{two directions}. 
This changes the analysis, but the main advantage is that we only need to consider boxes that are rectangles (i.e., with only four sides) in our dynamic program. 
This greatly simplifies the dynamic program, and we show that we need to consider only $f(\eps,\delta)\cdot (nN)^{\Oh(1)}$ different subproblems.
Hence, our EPTAS is both substantially faster and significantly simpler than previous work.

\medskip

Our second main contribution shows that several important subcases of {\sc{Maximum Weight Independent Set of Rectangles}} with $\delta$-shrinking admit a polynomial kernel when parameterized by $k$ and $\delta$. 
Intuitively, such a kernel is a polynomial-time algorithm that finds a subfamily of the input rectangles of size bounded by a polynomial of $k$ and $\delta$
that retains an optimum solution after $\delta$-shrinking; a formal definition is given in Section~\ref{sec:kern}.
We show that for unit
squares of non-uniform weight, we can construct a kernel of size $\Oh(k/\delta^2)$,
while for squares of non-uniform size, but of uniform weight, we can construct a kernel of size $\Oh(k^2\cdot \frac{\log (1/\delta)}{\delta^3})$.
As a direct consequence, we obtain FPT algorithms for the considered variants with running time $(k/\delta)^{\Oh(\sqrt{k})}\cdot (nN)^{\Oh(1)}$ 
by applying the $n^{\Oh(\sqrt{k})}$-time algorithm of Marx and Pilipczuk~\cite{MarxP15} on the kernels.
This \emph{subexponential} running time
is far better than the running time of the FPT algorithm for the general rectangle case.
Other results from the literature, like Alber and Fiala~\cite{AlberF04}, Hunt et al.~\cite{HuntMRRRS98}, or Chan et al.~\cite{chan2004note}, can be also used in combination with our kernelization to yield faster FPT and approximation algorithms.

The main idea of our kernelization algorithms is as follows.
If two or more squares have similar size and their centers lie very close to each other,
then in the shrinking model we can replace them by a single (shrunk)
square that is contained in each of them. 
In the weighted setting, the new square keeps the largest among the weights of the original ones.
Thus, we can assume that squares of roughly the same size do not lie too close
to each other. 
This enables a greedy approach, similar in spirit to Alber and Fiala~\cite{AlberF04}, which works for 
unit squares both in the uniform and non-uniform weight setting.
For squares of non-uniform size, but uniform weight, we observe that if two squares have very different size and they overlap,
then either one is contained in the other (and thus
we can remove the larger one without changing the optimum), or after the shrinking they
do not overlap anymore. Hence, if there are $k$ squares of pairwise very different
size, then they must be non-overlapping after shrinking, and we have a solution.
Otherwise, the squares ``occupy'' only $\Oh_\delta(k)$ levels of magnitude, and we can kernelize each of these levels separately.

\paragraph*{Organization.} In Section~\ref{sec:prelims} we establish the notation.
Section~\ref{sec:main} contains the presentation of these results, with the FPT algorithm presented as an adjustment of the EPTAS.
In Section~\ref{sec:kern} we explore kernelization for the considered problems; this section also contains the formal statements of the kernelization results.
Finally, in Section~\ref{sec:conc} we gather some concluding remarks and provide an outlook on future work.

\section{Preliminaries}\label{sec:prelims}

For clarity, we essentially adopt the notation of Adamaszek et al.~\cite{AdamaszekChalermsookWiese2015}.
Suppose that $\Rr=\{R_1,R_2,\ldots,R_n\}$ is a family of axis-parallel rectangles given in the input. Each rectangle $R_i$ is described as
$$R_i=\{(a,b)\ \colon\ \crdx{i}{1} <a<\crdx{i}{2} \textrm{ and } \crdy{i}{1}<b<\crdy{i}{2}\},$$
where $\crdx{i}{1}<\crdx{i}{2}$ and $\crdy{i}{1}<\crdy{i}{2}$ are integers. Thus, the input rectangles are assumed to be open, and their
vertices are at integral points. 
In general, this is not a restriction~\cite[Lemma~2.1]{LingasW05}.
We assume that the family $\Rr$ is given in the input with all the coordinates $\crdx{i}{1},\crdx{i}{2},\crdy{i}{1},\crdy{i}{2}$ encoded in binary;
thus, the coordinates are at most exponential in the total bit size of the input, denoted by $N$.
For a rectangle $R_i$, we define its {\em{width}} $\width_i=\crdx{i}{2}-\crdx{i}{1}$ and {\em{height}} $\height_i=\crdy{i}{2}-\crdy{i}{1}$.
Moreover, each rectangle $R_i$ has a prescribed {\em{weight}} $\weight_i$, which is a nonnegative real.
For a subset $\Ss\subseteq \Rr$, we denote $\weight(\Ss)=\sum_{R_i\in \Ss}\, \weight_i$.

Fix a constant $\delta$ with $0<\delta<1$. For a rectangle $R_i$, its {\em{$\delta$-shrinking}} $R_i^{-\delta}$ is the rectangle with $x$-coordinates 
$\crdx{i}{1}+\frac{\delta}{2}\width_i$ and $\crdx{i}{2}-\frac{\delta}{2}\width_i$, and $y$-coordinates $\crdy{i}{1}+\frac{\delta}{2}\height_i$ and $\crdy{i}{2}-\frac{\delta}{2}\height_i$.
The $\delta$-shrinking retains the weight $\weight_i$ of the original rectangle. For a subset $\Ss\subseteq \Rr$, we denote $\Ss^{-\delta}=\{R_i^{-\delta}\colon R_i\in \Ss\}$ to be the family of $\delta$-shrinkings
of rectangles from $\Ss$.

A family of rectangles is {\em{independent}} (or is an {\em{independent set}}) if the rectangles are pairwise non-overlapping.
In the \MWISR problem (\mwisr) we are given a family of axis-parallel rectangles $\Rr=\{R_1,R_2,\ldots,R_n\}$, and the goal is to find a subfamily of $\Rr$ that is independent and has  maximum total weight.
This maximum weight will be denoted by $\OPT(\Rr)$. In the parameterized setting, the {\em{parameterized \mwisr}} problem, we are additionally given an integer parameter~$k$, and we look for a subfamily of $\Rr$ that has size at most $k$, 
is independent, and has maximum possible weight subject to these conditions. This maximum weight will be denoted by $\OPT_k(\Rr)$.

In the $\delta$-shrinking setting, we relax the requirement of independence to just requiring the disjointness of $\delta$-shrinkings,
but we still compare the weight of the output of our algorithm with $\OPT(\Rr)$, respectively with $\OPT_k(\Rr)$. 

Finally, given a family $\Rr$ of arbitrary objects in the plane, the \emph{intersection graph} $G$ induced by $\Rr$ has a vertex for each object in $\Rr$ and an edge between two vertices if and only if the corresponding objects intersect. Conversely, the family $\Rr$ is said to be a \emph{representation} of $G$.

\section{Main results}\label{sec:main}

With the above definitions in mind, we can state formally our main results.

\begin{theorem}[FPT for \mwisr with $\delta$-shrinking]\label{thm:main-fpt}
There is a randomized algorithm that given a weighted family $\Rr$ of $n$ axis-parallel rectangles with total encoding size $N$, 
and parameters $k$ and $\delta$, runs in time $f(k,\delta)\cdot (nN)^{c}$ for some computable function
$f$ and constant $c$, and outputs a subfamily $\Ss\subseteq \Rr$ such that $|\Ss|\leq k$, $\Ss^{-\delta}$ is independent, and $\weight(\Ss)\geq \OPT_k(\Rr)$ with probability at least~$1/2$.
\end{theorem}

\begin{theorem}[EPTAS for \mwisr with $\delta$-shrinking]\label{thm:main-approx}
There is a randomized algorithm that given a weighted family $\Rr$ of $n$ axis-parallel rectangles with total encoding size $N$, 
and parameters $\delta,\eps$, runs in time $f(\eps,\delta)\cdot (nN)^{c}$ for some computable function
$f$ and constant $c$, and outputs a subfamily $\Ss\subseteq \Rr$ such that $\Ss^{-\delta}$ is independent, and $\weight(\Ss)\geq (1-\eps)\OPT(\Rr)$ with probability at least~$1/2$.
\end{theorem}

We prove the above thereoms in this section.
We concentrate on the proof of Theorem~\ref{thm:main-approx}, because Theorem~\ref{thm:main-fpt} follows by a simple adjustment of our argumentation, as we will explain
in Section~\ref{sec:param}.

Throughout the proof we fix the input family $\Rr=\{R_1,\ldots,R_n\}$, and we denote $\OPT=\OPT(\Rr)$.
We also fix the constants $\delta$ and $\eps$, and w.l.o.g.
we assume that $1/\delta$ and $1/\eps$ are even integers larger than~$4$. 
For convenience, throughout the proof we aim at finding a solution $\Ss$ with $\weight(\Ss)\geq (1-d\cdot \eps)\OPT$ for some constant $d$, for at the end we may rescale the parameter $\eps$ to $\eps/d$.

Recall that $N$ denotes the total size of the binary encoding of all the coordinates of the input rectangles.
By shifting all the rectangles, we may assume without loss of generality that they all fit into the square $[1,L]\times [1,L]$, where $L=(1/\delta\eps)^{\ellmax}$ for some integer $\ellmax=\Oh(N)$.
That is, all the coordinates $\crdx{i}{1},\crdx{i}{2},\crdy{i}{1},\crdy{i}{2}$ are between $1$ and $L$, so in particular the width and the height of each rectangle is smaller than $L$.

Our reasoning is divided into two main steps.
First, like in the PTAS due to Adamaszek et al.~\cite{AdamaszekChalermsookWiese2015},
in Section~\ref{sec:sparse} we use standard shifting arguments to remove some rectangles from $\Rr$ so that $\OPT$ decreases only by an $\Oh(\eps)$-fraction, but the resulting family admits 
some useful properties. Then, we shrink the rectangles in a similar way as in~\cite{AdamaszekChalermsookWiese2015}. As we will point out, there is a subtle but important 
difference in our shrinking procedure compared to~\cite{AdamaszekChalermsookWiese2015}.
Second, we show that the properties of the obtained family enable us to compute an optimum solution
using dynamic programming; this algorithm is presented in Section~\ref{sec:dp}.

\subsection{Sparsifying the family}\label{sec:sparse}

\paragraph*{\ej{Separating rectangles by size.}}
Intuitively, we will apply shifting techniques to extract some structure in the input family $\Rr$ while losing only an $\Oh(\eps)$-fraction of $\OPT$.
The first goal is to obtain a classification of the rectangles according to their widths and heights, such that rectangles in the same class have similar widths and heights, but between the classes
the dimensions differ significantly. This is encapsulated formally in the following definition.

\begin{definition}\label{def:well-separated}
A subfamily $\Rr'\subseteq \Rr$ is {\em{well-separated}} if there exist two partitions 
$$(\Rr^{\Vf}_1,\Rr^{\Vf}_2,\ldots,\Rr^{\Vf}_p)\quad \textrm{and} \quad (\Rr^{\Hf}_{1},\Rr^{\Hf}_2,\ldots,\Rr^{\Hf}_p)$$
of $\Rr'$, with $p\leq \ellmax$, as well as reals $\nu_t,\mu_t$ for $t=1,2,\ldots,p$, with the following properties satisfied for each $t\in \{1,2,\ldots,p\}$:
\begin{itemize}
\item $\nu_{t}\le \width_{i}<\mu_{t}\ $ for each $R_{i}\in\Rr^{\Vf}_{t}$;
\item $\nu_{t}\le \height_{i}<\mu_{t}\ $ for each $R_{i}\in\Rr^{\Hf}_{t}$;
\item $\nu_{t}/\mu_{t-1}=1/\delta\epsilon$ (except for $t=1$) and $\mu_{t}/\nu_{t}= (1/\delta\epsilon)^{(1/\epsilon)-1}$; and 
\item $\nu_1\leq 1$, $\mu_p\geq L$, and all numbers $\mu_t$ and $\mu_t$ apart from $\nu_1$ are integers.
\end{itemize}
\end{definition}

The partitions $(\Rr^{\Vf}_t)_{t=1,\ldots,p}$ and $(\Rr^{\Hf}_t)_{t=1,\ldots,p}$ are called the {\em{vertical}} and {\em{horizontal levels}}, respectively, whereas the
parameters $(\nu_t)_{t=1,\ldots,p}$ and $(\mu_t)_{t=1,\ldots,p}$ are the {\em{lower}} and {\em{upper limits}} of the corresponding levels.
Note that, maybe somehow counter-intuitively,
vertical levels partition \ej{$\Rr'$ by}
width, while the horizontal levels partition \ej{$\Rr'$ by}
heights; the reason for this will become clear later on.  
We now prove that we can find a well-separated subfamily that loses only an $\Oh(\eps)$-fraction of $\OPT$ using a standard shifting technique (see, e.g., Hochbaum and Maas~\cite{HochbaumM85}).
Essentially the same step is used in the PTAS of Adamaszek et al.~\cite{AdamaszekChalermsookWiese2015} (see Lemma~6 therein).

\begin{lemma}\label{lem:well-separated}
In polynomial time one can sample a subfamily $\Rr'\subseteq \Rr$ that is well-separated and satisfies $\OPT(\Rr')\geq (1-8\eps)\OPT$ with probability at least $3/4$.
Moreover, $\Rr'$ is constructed together with the corresponding partitions into vertical and horizontal levels, and with their lower and upper limits.
\end{lemma}
\begin{proof}
Recall that the widths and heights of the rectangles from $\Rr$ are integers between $1$ and $L-1$, where $L=(1/\delta\eps)^{\ellmax}$.
Create first a partition of the rectangles into {\em{vertical layers}} $\Ll^{\Vf}_j$ for $j=1,2,\ldots,\ellmax$, where layer $\Ll^{\Vf}_j$ consists of
rectangles $R_i$ for which $(1/\delta\eps)^{j-1}\leq \width_i<(1/\delta\eps)^j$.
In a symmetric manner, partition $\Rr$ into {\em{horizontal layers}} $\Ll^{\Hf}_j$ for $j=1,2,\ldots,\ellmax$, where layer $\Ll^{\Hf}_j$ consists of
rectangles $R_i$ for which $(1/\delta\eps)^{j-1}\leq \height_i<(1/\delta\eps)^j$.

Let us pick an offset $b$ uniformly at random from the set $\{0,1,\ldots,(1/\eps)-1\}$.
Construct $\Rr'$ from $\Rr$ by removing all rectangles contained in those vertical layers $\Ll^{\Vf}_j$ and those horizontal layers $\Ll^{\Hf}_j$, for which $j\equiv b \mod (1/\eps)$.
It is easy to see that $\Rr'$ constructed in this manner is well-separated: each vertical level $\Rr^{\Vf}_t$ consists of $(1/\eps)-1$ consecutive vertical layers between two removed ones,
with the exception of the first and the last level, for which the start/end of the sequence of layers delimits the level.
A symmetric analysis yields the partition into horizontal levels. Moreover, it is straightforward to compute in polynomial time 
the partition into horizontal/vertical levels, as well as to choose their lower and upper limits.


It remains to show that the probability that $\OPT(\Rr')\geq (1-8\eps)\Rr$ is at least $3/4$.
Fix any optimum solution $\Ss$ in $\Rr$, that is, an independent set of rectangles such that $\weight(\Ss)=\OPT$.
Observe that for any rectangle $R_i\in \Ss$, the probability that the vertical layer it belongs to is removed during the construction of $\Rr'$, is equal to $\eps$.
Similarly, the probability that the horizontal layer to which $R_i$ belongs is removed when constructing $\Rr'$, is also $\eps$.
Hence, $R_i$ is not included in $\Rr'$ with probability at most $2\eps$. 
This means that the expected value of $\weight(\Ss\setminus \Rr')$, the total weight of rectangles from $\Ss$ that did not survive in $\Rr'$, is at most $2\eps\cdot \OPT$.
By Markov's inequality, with probability at least $3/4$ we have that $\weight(\Ss\setminus \Rr')\leq 8\eps\cdot \OPT$.
This event, however, implies that $\OPT(\Rr')\geq \weight(\Ss\cap \Rr')\geq (1-8\eps)\OPT$, and hence we are done.
\end{proof}

From now on we work with the family $\Rr'$ obtained by Lemma~\ref{lem:well-separated}, and we adopt the notation from Definition~\ref{def:well-separated}.

\paragraph*{Hierarchical grid structure.} Let $a \ej{ \in \{1,\ldots,L-1\}}$ be an integer shift parameter, to be determined later.
Given $a$, we construct a hierarchy of grid lines in the plane.
We start with horizontal lines, which will be divided into $p$ levels corresponding to the horizontal levels $\Rr^{\Hf}_t$.
For level $t$, define the {\em{level-$t$ unit}} as $u_t=\delta\nu_t/2$. 
Note that for $t>1$ we have $u_t=\mu_{t-1}/(2\eps)$, and hence $u_t$ is an integer for $t>1$, since $1/\eps$ is even.
For each level $t\in \{1,2,\ldots,p\}$ we define a set of horizontal grid lines $G^{\Hf}_t$, consisting of the horizontal lines with $y$-coordinates from the set
$$\{a+b\cdot u_t\colon b\in \mathbb{Z}\}.$$
In other words, we take horizontal lines that are $u_t$ apart from each other, and we shift them so that there is a line with $y$-coordinate $a$.
We define vertical grid lines $G^{\Vf}_t$ of levels $t=1,2,\ldots,p$ in a symmetric manner, using the same shift parameter $a$ and the same units for all levels. 
Define the {\em{grid}} of level $t$ to be $G_t=G^\Hf_t\cup G^\Vf_t$.
Note that any line of $G_t$ is also a line of $G_{t'}$ for all $t'<t$.
Thus, the grid of each level $t'$ refines the grid of each larger level $t$.

Before we proceed, we briefly describe the intuition of the next step;
this step is also present in the PTAS of Adamaszek et al.~\cite{AdamaszekChalermsookWiese2015} (see Lemma~7 therein). Rectangles belonging to $\Rr^\Vf_{t'}$ for $t'\geq t$ have width not smaller
than $\nu_t$. On the other hand, the lines of $G^\Vf_t$ are spaced at distance $u_t=\delta \nu_t/2$ apart, which means that there are $\Omega(1/\delta)$ vertical grid lines of $G^\Vf_t$ crossing
each rectangle of vertical level $t$ or larger. Intuitively, $G^\Vf_t$ provides a fine grid for those vertical levels, so that their rectangles can be snapped to the lines of $G^\Vf_t$ via shrinking by a 
multiplicative factor of at most $1-\delta$. On the other hand, the rectangles of vertical levels $t-1$ or smaller have widths not larger than $\mu_{t+1}=\nu_t\cdot (\delta\eps)$.
This means that the grid lines of $G^\Vf_t$ are actually at much larger distance from each other than the maximum possible width of such rectangles; more precisely, larger by a multiplicative factor at least $1/(2\eps)$.
Consequently, if we choose the shift parameter $a \ej{ \in \{1,\ldots,L-1\}}$ uniformly at random, the probability that 
\ej{a}
rectangle $R_i$ will be crossed by a vertical line of 
level larger than its vertical level, or a horizontal line of level larger than its horizontal level,
will be $\Oh(\eps)$. If we exclude such rectangles, then we lose only an $\Oh(\eps)$-fraction of $\OPT$ in expectation, 
while achieving the property that vertical lines of each level $t$ separate rectangles of lower vertical levels,
and the symmetric claim holds for horizontal lines as well. 

We now formalize the above intuition.
Take a rectangle $R_i\in \Rr'$, and suppose that $R_i\in \Rr^\Vf_{s}$ and $R_i\in \Rr^\Hf_t$.
We say that $R_i$ is {\em{abusive}} if $R_i$ is crossed by a vertical line of level larger than $s$, or $R_i$ is crossed by a horizontal line of level larger than~$t$; see Figure~\ref{fig:abusive}.

\begin{lemma}\label{lem:abusive}
Suppose that $a$ is sampled uniformly at random from $\{0,\ldots,L-1\}$. Then, for each $R_i\in \Rr'$, the probability that $R_i$ is abusive is at most $8\eps$.
\end{lemma}
\begin{proof}
Suppose $R_i\in \Rr^\Vf_{s}\cap \Rr^\Hf_t$.
Rectangle $R_i$ is crossed by a vertical line of level larger than $s$ if and only if it is crossed by a vertical line of level $s+1$.
Since $R_i\in \Rr^\Vf_{s}$, we have that 
$$\width_i\leq \mu_s\leq \nu_{s+1}\cdot (\delta\eps)=u_{s+1}\cdot 2\eps.$$
Lines of $G^\Vf_{s+1}$ are spaced at distance $u_{s+1}$ from each other, which means that $R_i$ is crossed by a line of $G^\Vf_{s+1}$ if and only if
the remainder of $a$ modulo $u_{s+1}$ is among a set of $\width_i-1$ values, corresponding to the vertical lines with integral $x$-coordinates that cross $R_i$.
As $u_{s+1}<L$ and $a$ is drawn uniformly at random from $\{0,1,\ldots,L-1\}$, this happens with probability at most $2\cdot \frac{\width_i-1}{u_{s+1}}\leq 4\eps$.
Symmetrically, $R_i$ is crossed by a horizontal line of level larger than $t$ with probability at most $4\eps$.
Therefore, $R_i$ is abusive with probability at most $8\eps$.
\end{proof}

\begin{figure}[t]
\begin{centering}
\includegraphics[width=0.8\textwidth]{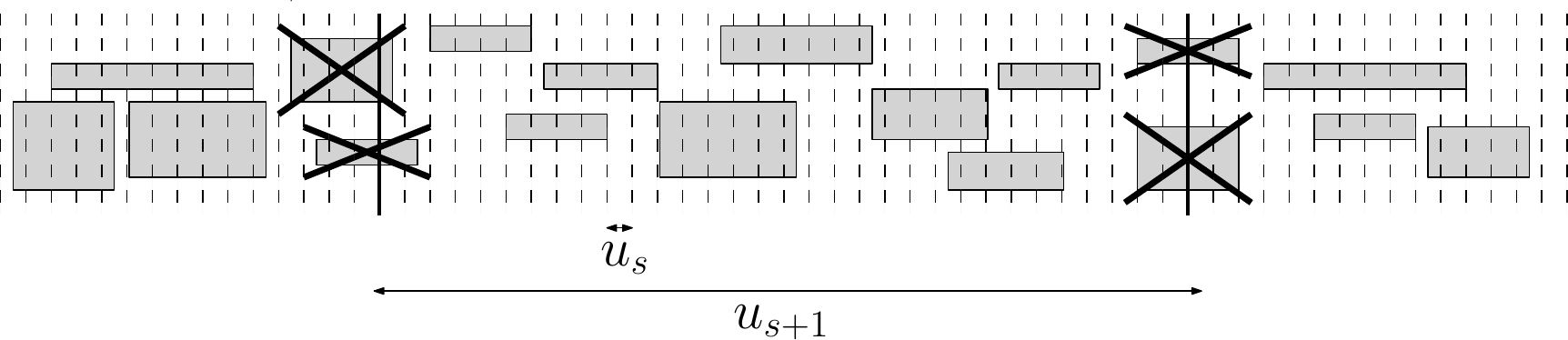}
\caption{\label{fig:abusive}The vertical grid. The dashed vertical lines are the vertical grid lines of $G^\Vf_s$, the bold vertical lines are the lines in the set
$G^{\Vf}_{s+1}$. All shown rectangles are from level $\Rr^\Vf_s$.
The crossed out rectangles are abusive since they intersect lines
from $G^\Vf_{s+1}$.}
\end{centering}
\end{figure}

We now construct a subfamily $\Rr''\subseteq \Rr'$ by drawing $a$ uniformly at random from the set $\{0,1,\ldots,L-1\}$, defining the
grid lines $(G^\Vf_t)_{t=1,\ldots,p}$ and $(G^\Hf_t)_{t=1,\ldots,p}$, and removing all abusive rectangles from $\Rr'$. Lemma~\ref{lem:abusive} implies the following. 

\begin{lemma}\label{lem:abusive-prob}
$\OPT(\Rr'')\geq (1-32\eps)\OPT(\Rr')$ with probability at least $3/4$.
\end{lemma}
\begin{proof}
Let $\Ss\subseteq \Rr'$ be such that $\weight(\Ss)=\OPT(\Rr')$.
Since each rectangle of $\Rr'$ is abusive with probability at most $8\eps$, we have that the expected total weight of rectangles of $\Ss$ that turns out to be abusive is at most $8\eps\cdot \OPT(\Rr')$.
By Markov's inequality, the probability that this total weight exceeds $32\eps\cdot \OPT(\Rr')$ is at most $1/4$.
Consequently, with probability at least $3/4$ we have that $\OPT(\Rr'')\geq \weight(\Ss\cap \Rr'')\geq (1-32\eps)\OPT(\Rr')$.
\end{proof}

Since the sampling of $a$ is done independently of the sampling used for the construction of $\Rr'$, Lemmas~\ref{lem:well-separated} and~\ref{lem:abusive-prob} combined imply that
\begin{equation}\label{eq:whole-loss}
\OPT(\Rr'')\geq (1-32\eps)(1-8\eps)\OPT(\Rr)\geq (1-40\eps)\OPT(\Rr),
\end{equation}
with probability at least $3/4\cdot 3/4>1/2$.
Moreover, $\Rr''$ is still well-separated, where the horizontal and vertical levels are $\Rr^\Hf_t$ and $\Rr^\Vf_t$ trimmed to $\Rr''$, and their limits are unchanged.
Finally, $\Rr''$ does not contain any abusive rectangle.

\paragraph*{Snapping by shrinking.}
When considering $\Rr''$, the vertical grid lines of $G^\Vf_t$ provide a fine division of every rectangle from vertical level $t$ or larger, while no rectangle of smaller vertical level is crossed by them;
a symmetric claim holds also for horizontal grid lines. The idea now is to shrink each rectangle $R_i\in \Rr''$ so that its vertical sides are aligned with some vertical grid lines of the vertical level of $R_i$,
while the horizontal sides are aligned with some horizontal grid lines of the horizontal level of $R_i$. This is formalized in the next lemma, which is also similar to Adamaszek et al.~\cite{AdamaszekChalermsookWiese2015}. However, in this step there is a
subtle but crucial difference.
Consider a rectangle $R_i\in \Rr''$, and suppose that $R_i\in \Rr^\Vf_{s}$ and $R_i\in \Rr^\Hf_t$. In \cite{AdamaszekChalermsookWiese2015} the rectangle $R_i$ is shrunk in the vertical dimension only if $s\ge t$ and in the horizontal dimension only if $t\ge s$. In this paper we always do both which will be important for our dynamic programming algorithm later.

\begin{lemma}\label{lem:snapping}
In polynomial time we can compute a family of axis-parallel rectangles $\Qq$ that contains one rectangle $Q_i$ for each $R_i\in \Rr''$, of the same weight $\weight_i$ as $R_i$ and 
satisfying the following conditions:
\begin{itemize}
\item $R_i^{-\delta}\subseteq Q_i\subseteq R_i$ for each $R_i\in \Rr''$; and
\item if $R_i\in \Rr^\Vf_{s}\cap \Rr^\Hf_t$, then both vertical sides of $Q_i$ are contained in some vertical grid lines of $G^\Vf_s$, and both horizontal sides of $Q_i$ are contained in some horizontal grid lines of $G^{\Hf}_t$.
\end{itemize}
\end{lemma}
\begin{proof}
Take any $R_i\in \Rr''$, and suppose $R_i\in \Rr^\Vf_{s}\cap \Rr^\Hf_t$.
We define $Q_i$ as the rectangle cut from the plane by the following four lines: 
\begin{itemize}
\item the left-most and the right-most vertical grid lines of $G^\Vf_s$ that cross $R_i$;
\item the bottom-most and the top-most horizontal grid lines of $G^\Hf_t$ that cross $R_i$.
\end{itemize}
Clearly, we have that $Q_i\subseteq R_i$ and the second condition of the statement is satisfied.
We are left with proving that $R_i^{-\delta}\subseteq Q_i$.

Consider first the left side of $Q_i$, which is contained in the left-most vertical grid line of $G^\Vf_s$ that crosses $R_i$.
Since $R_i\in \Rr^{\Vf}_s$, we have that $\width_i\geq \nu_s$, while the grid lines of $G^\Vf_s$ are spaced at distance $u_t=\delta\nu_s/2$ apart.
This means that the left-most vertical grid line crossing $R_i$ has the $x$-coordinate not larger than $\crdx{i}{1}+\delta\nu_s/2$, which in turn is not larger than $\crdx{i}{1}+\delta g_i /2$.
This means that the left side of $Q_i$ is either to the left or at the same $x$-coordinate as the left side of $R_i^{-\delta}$.
An analogous reasoning can be applied to the other three sides of $Q_i$, thereby proving that $R_i^{-\delta}\subseteq Q_i$.
\end{proof}

In the next subsection we will actually show that we can find an optimum solution for the family $\Qq$ using dynamic programming within the claimed running time.
By~(\ref{eq:whole-loss}) and the first condition of Lemma~\ref{lem:snapping}, we have
$$\OPT(\Qq)\geq \OPT(\Rr'')\geq (1-40\eps)\OPT$$
with probability at least $1/2$. Hence, the optimum solution for $\Qq$ indeed has large enough weight with good probability.
Moreover, by the first condition of Lemma~\ref{lem:snapping}, for any independent set of rectangles in $\Qq$, the corresponding rectangles in $\Rr^{-\delta}$ are also independent.
Hence, any solution for \mwisr on $\Qq$ projects to a solution of the same weight for \mwisr with $\delta$-shrinking on $\Rr$.

From now on we focus on the family $\Qq$. For each $t\in \{1,2,\ldots,p\}$, let 
$$\Qq^\Vf_t=\{Q_i\colon R_i\in \Rr^\Vf_t\}\quad \textrm{and}\quad \Qq^\Hf_t=\{Q_i\colon R_i\in \Rr^\Hf_t\}.$$
Note that since $\Qq$ is obtained only by shrinking rectangles from $\Rr''$, it is still the case that no rectangle of $\Qq$ is abusive.

\subsection{Dynamic programming}\label{sec:dp}

We now present a dynamic programming algorithm that, given the family $\Qq$ constructed in the previous section, computes the value $\OPT(\Qq)$.
An optimum solution, that is, an independent set of rectangles of total weight $\OPT(\Qq)$, can be recovered from the run of the dynamic program using standard methods, 
and hence for simplicity we omit this aspect in the description.

Actually, it will be convenient to describe the algorithm as {\em{backtracking with memoization}}.
That is, the subproblems are solved by recursion, but once a subproblem has been solved once, the optimum value for it is stored in a map (is {\em{memoized}}), and further calls to solving this
subproblem will only retrieve the memoized optimum value, rather than solve the subproblem again. Solving each subproblem (excluding the recursive subcalls) will take time $f(\delta,\eps)\cdot n^{\Oh(1)}$ 
for some computable function $f$, and we will argue that at most $g(\delta,\eps)\cdot (nN)^{\Oh(1)}$ subproblems are solved in total, for some other computable function $g$. 
This will ensure the promised running time of the algorithm.

We first define subproblems. A {\em{subproblem}} is a tuple $I=(s,t,x_1,x_2,y_1,y_2)$:
\begin{itemize}
\item the pair $(s,t)\in \{1,\ldots,p\}\times \{1,\ldots,p\}$ is the {\em{level}} of the subproblem, which consists of the vertical level $s$ and the horizontal level $t$; and
\item $x_1,x_2,y_1,y_2$ are integers satisfying 
$$x_1< x_2\leq x_1+(1/\delta\eps)^{1/\eps}\quad\textrm{and}\quad y_1< y_2\leq y_1+(1/\delta\eps)^{1/\eps}.$$
Integers $x_1,x_2$ are the lower and upper {\em{vertical offsets}}, respectively, while $y_1,y_2$ are the lower and upper {\em{horizontal offsets}}, respectively. 
\end{itemize}
The {\em{area covered}} by subproblem $I=(s,t,x_1,x_2,y_1,y_2)$ is the rectangle
$$A_I=(a+x_1\cdot u_s,a+x_2\cdot u_s)\times (a+y_1\cdot u_t,a+y_2\cdot u_t).$$
In other words, $(x_1,x_2,y_1,y_2)$ define the offsets of the four grid lines---two from $G^\Vf_s$ and two from $G^\Hf_t$---that cut out $A_I$ from the plane.

For a subproblem $I$, let $\Qq_I$ be the set of all rectangles from $\Qq$ that are contained in $A_I$.
The next check follows from a simple calculation of parameters.

\begin{lemma}\label{lem:subproblem-level}
If $I$ is a subproblem of level $(s,t)$, then $\Qq_I\subseteq \bigcup_{s'\leq s,\, t'\leq t} \Qq^\Vf_{s'}\cap \Qq^\Hf_{t'}$.
\end{lemma}
\begin{proof}
Consider a rectangle $Q_i$ that belongs to a vertical level $\Qq^\Vf_{s'}$ for some $s'>s$.
By the first condition of Lemma~\ref{lem:snapping}, we have that the width of $Q_i$ is at least $(1-\delta)\width_i$.
However, since $s'>s$, we have $\width_i\geq \nu_{s'}\geq \nu_s\cdot (1/\delta\eps)^{1/\eps}$.
On the other hand, the width of $A_I$ is at most $(1/\delta\eps)^{1/\eps}\cdot u_s=\nu_s\cdot (1/\delta\eps)^{1/\eps}\cdot (\delta/2)$.
Since we assumed that $1/\delta>4$, we infer that the width of $Q_i$ is larger than the width of $A_I$, and hence $Q_i$ cannot be contained in $A_I$.
A symmetric reasoning shows also that any rectangle that belongs to a horizontal level $\Qq^\Hf_{t'}$ for some $t'>t$ cannot belong to $\Qq_I$ due to having larger height than $A_I$.
The claim follows.
\end{proof}

For a subproblem $I$, we define the {\em{value}} of $I$, denoted $\Val(I)$, as follows:
$$\Val(I)=\max\, \{\, \weight(\Ss)\ \colon\ \Ss\subseteq \Qq_I \textrm{ and }\Ss\textrm{ is independent}\, \}.$$
We first note that there is a subproblem that encompasses the whole instance.

\begin{lemma}\label{lem:global-subproblem}
There is a subproblem $\glob$ of level $(p,p)$, computable in constant time, such that $A_{\glob}\supseteq (1,L)\times (1,L)$. Consequently, $\OPT(\Qq)=\Val(\glob)$.
\end{lemma}
\begin{proof}
Choose $z_1$ to be the largest integer such that $a+z_1\cdot u_p\leq 0$.
Let $z_2=z_1+(1/\delta\eps)^{1/\eps}$, and define $\glob=(p,p,z_1,z_1,z_2,z_2)$.
Clearly $\glob$ is a subproblem, so we are left with proving that $A_{\glob}\supseteq (1,L)\times (1,L)$.
By the maximality of $z_1$,
\begin{equation}\label{eq:z1}
a+z_1\cdot u_p>-u_p.
\end{equation}
On the other hand, we have
\begin{equation}\label{eq:up}
u_p=\delta\nu_p/2=\mu_p\cdot (\delta/2)\cdot (\delta\eps)^{1/\eps-1}\geq L\cdot \delta/2\cdot (\delta\eps)^{1/\eps-1}.
\end{equation}
By combining (\ref{eq:z1}) and~(\ref{eq:up}), we obtain
\begin{eqnarray*}
a+z_2\cdot u_p & = & a+z_1\cdot u_p+(1/\delta\eps)^{1/\eps}\cdot u_p\geq -u_p + (1/\delta\eps)^{1/\eps}\cdot u_p\\
& \geq & L\cdot \delta/2\cdot (\delta\eps)^{1/\eps-1}\cdot ((1/\delta\eps)^{1/\eps}-1) > L;
\end{eqnarray*}
the last inequality follows from the assumption that $1/\eps$ is an integer larger than~$4$.
It follows that $A_\glob$ indeed contains the whole square $(1,L)\times (1,L)$.
\end{proof}

Next, we show how to {\em{solve}} each subproblem $I$, that is, to compute $\Val(I)$, using recursion.

\begin{lemma}\label{lem:recursion}
A subproblem $I$ of level $(s,t)$ can be solved using $f(\delta,\eps)$ calls to solving subproblems of levels $(s-1,t)$, $(s,t-1)$, and $(s-1,t-1)$,
for some computable function $f$.
Moreover, the time needed for this computation, excluding the time spent in the recursive calls, is at most $f(\delta,\eps)\cdot n$.
\end{lemma}
\begin{proof}
Let $I=(s,t,x_1,x_2,y_1,y_2)$. We denote $x=x_2-x_1$ and $y=y_2-y_1$; recall that $0<x,y\leq (1/\delta\eps)^{1/\eps}$.
Consider all vertical grid lines of $G^\Vf_s$ with $x$-coordinates $a+\alpha\cdot u_s$ for $x_1\leq \alpha\leq x_2$ and
all horizontal grid lines of $G^\Hf_t$ with $y$-coordinates $a+\beta\cdot u_t$ for $y_1\leq \beta\leq y_2$.
These lines partition the rectangle $A_I$ into $x\times y$ smaller rectangles with side lengths $u_s$ and $u_t$, which we shall call {\em{cells}}.
Cells can be naturally indexed by pairs of integers from $\{1,\ldots,x\}\times \{1,\ldots,y\}$:
cell $(\alpha,\beta)$ has sides contained in vertical lines with $x$-coordinates $a+(\alpha-1)\cdot u_s$ and $a+\alpha\cdot u_s$,
and horizontal lines with $y$-coordinates $a+(\beta-1)\cdot u_t$ and $a+\beta\cdot u_t$.
We use the terminology of {\em{rows}} and {\em{columns}} in the grid of cells, defined naturally.

Consider any independent set of rectangles $\Ss\subseteq \Qq_I$. The reader may think of $\Ss$ as an optimum solution, that is, a set with $w(\Ss)=\Val(I)$.
By Lemma~\ref{lem:global-subproblem}, each rectangle of $\Ss$ belongs to $\Qq^\Vf_{s'}\cap \Qq^\Hf_{t'}$ for some $s'\leq s$ and $t'\leq t$; such rectangle will be henceforth called a rectangle of {\em{level}} $(s',t')$.
Let us classify these rectangles according to their level as follows. A rectangle $Q_i\in \Ss$ of level $(s',t')$ is called:
\begin{itemize}
\item {\em{large}} if $s'=s$ and $t'=t$;
\item {\em{horizontal}} if $s'=s$ and $t'<t$;
\item {\em{vertical}} if $s'<s$ and $t'=t$;
\item {\em{small}} if $s'<s$ and $t'<t$.
\end{itemize}
Observe that if $s=1$ then there are no vertical or small rectangles, whereas if $t=1$ then there are no horizontal or small rectangles.
The following assertions follow from the second condition of Lemma~\ref{lem:snapping}, and the fact that no rectangle of $\Qq$ is abusive; see Figure~\ref{fig:divide} for an illustration.
\begin{itemize}
\item Each large rectangle is equal to the union of a subset of cells.
\item Each horizontal rectangle is contained in the union of a sequence of consecutive cells in a single row of the grid. 
The vertical sides of this rectangle are aligned with the left side of the first cell of the sequence, and the right side of the last cell.
\item Each vertical rectangle is contained in the union of a sequence of consecutive cells in a single column of the grid. 
The horizontal sides of this rectangle are aligned with the bottom side of the first cell of the sequence, and the top side of the last cell.
\item Each small rectangle is contained in a single cell.
\end{itemize}
We shall call a cell:
\begin{itemize}
\item {\em{large}} if it is contained in a large rectangle;
\item {\em{horizontal}} if it overlaps with some horizontal rectangle;
\item {\em{vertical}} if it overlaps with some vertical rectangle; and
\item {\em{small}} if all rectangles from $\Ss$ with which it overlaps are small.
\end{itemize}
Observe that as $\Ss$ is independent, each cell is of exactly one of the above types.

We now partition the cells into {\em{boxes}}, where each box is formed by a rectangle of cells; again, see Figure~\ref{fig:divide} for an illustration.
First, for each large rectangle, create a {\em{large box}} consisting of the cells contained in this rectangle.
Second, for every inclusion-wise maximal sequence of consecutive horizontal cells contained in the same row, create a {\em{horizontal box}} consisting of all these cells.
Third, for every inclusion-wise maximal sequence of consecutive vertical cells contained in the same column, create a {\em{vertical box}} consisting of all these cells.
Finally, put every small cell into a {\em{small box}} consisting only of this cell.
Observe that thus, every rectangle of $\Ss$ is contained in a single box.

Suppose for a moment that $s>1$ and consider any potential vertical box $B$, that is, a sequence of consecutive cells contained in one column of the grid of cells.
Observe that for this box we can define a subproblem $I_B$ of level $(s-1,t)$ so that $B=A_{I_B}$.
This is because the grid lines delimiting $B$ are also grid lines from $G^{\Vf}_{s-1}$ and $G^{\Hf}_{t}$, while the width of $B$ is equal to
$u_s=(1/\delta\eps)^{1/\eps} u_{s-1}$. 
Hence, for the subproblem $I_B$ the difference between its vertical offsets will be bounded by $(1/\delta\eps)^{1/\eps}$,
as requested in the definition of a subproblem.
Similarly we can define a subproblem $I_B$ of level $(s,t-1)$ whenever $B$ is a horizontal box, and a subproblem $I_B$ of level $(s-1,t-1)$ whenever $B$ is a small box.

\begin{figure}[t]
\begin{centering}
\includegraphics[scale=0.6]{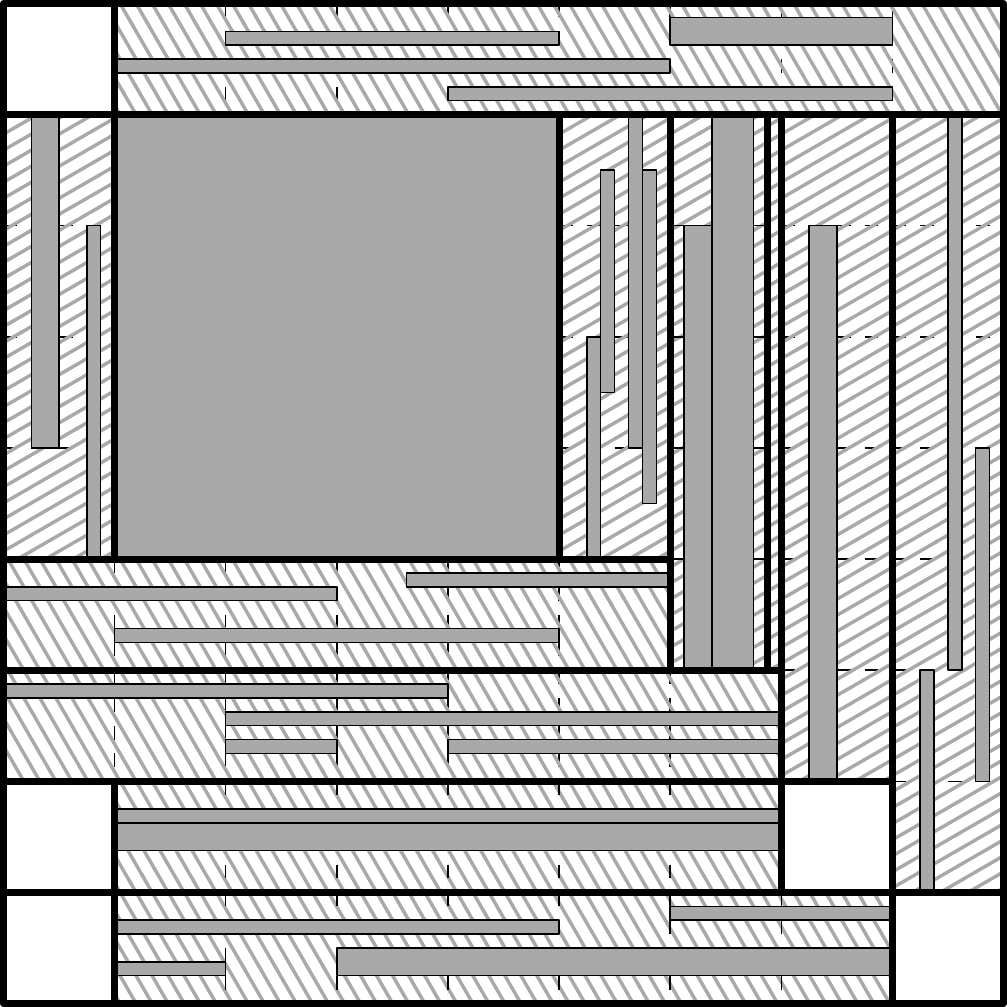}
\caption{The partition of $A_I$ into large, horizontal, vertical, and small cells. The figure is a slightly modified figure from~\cite{AdamaszekChalermsookWiese2015}.}\label{fig:divide}
\end{centering}
\end{figure}

With these observations in mind, we can present the procedure for computing $\Val(I)$; suppose for now that $s,t>1$.
The procedure enumerates
all partitions of the grid of cells into boxes, where each box has a prescribed type: it is either large, horizontal, vertical, or small.
Obviously, we consider only partitions where horizontal boxes are contained in one row, vertical boxes are contained in one column, and small boxes consist of one cell.
Observe that since $x,y\leq (1/\delta\eps)^{1/\eps}$, the number of such partitions is bounded by $f(\delta,\eps)$ for some computable function $f$.
For each partition $\Pp$, we do the following:
\begin{itemize}
\item For each large box $B$ in $\Pp$, choose the heaviest rectangle $Q_i\in \Qq_I$ such that $B=Q_i$, and record its weight. If for some large box there is no such rectangle, ignore this partition $\Pp$.
\item For each horizontal, vertical, or small box $B$ in $\Pp$, compute $\Val(I_B)$ by a recursive call to solving the subproblem $I_B$.
\item Sum up the weights recorded for large boxes and values computed for other boxes, and record the obtained sum as a candidate value for $\Val(I)$.
\end{itemize}
Finally, return as $\Val(I)$ the maximum among all the candidate values computed for all partitions into boxes.
In case $s=1$ we assume that all vertical and small boxes yield value $0$, and when $t=1$ we do the same for all horizontal and small boxes.
The claimed running time of the procedure, as well as the upper bound on the number of recursive calls, follows directly from the description.

To see the correctness of the procedure, observe that each candidate value represents the total weight of some independent set of rectangles contained in $\Qq_I$:
this independent set consists of all the rectangles chosen to fill the large boxes, and the union of the independent sets yielding values $\Val(I_B)$ among all the other boxes $B$.
On the other hand, from the discussion before the description of the procedure it follows that if $\Ss_I\subseteq \Qq_I$ is an independent set of rectangles such that $\weight(\Ss_I)=\Val(I)$,
then there will be a partition $\Pp$ into boxes which will be considered by the procedure, and for which every rectangle of $\Ss_I$ fits into a single box.
It is then easy to see that for this partition $\Pp$, the optimum value $\Val(I)$ will be computed as a candidate.
\end{proof}

Finally, to bound the running time of the algorithm, we prove that there is only a small number of subproblems $I$ for which $\Qq_I$ is non-empty.
Obviously, only such subproblems are necessary to solve, as the others have value~$0$.

\begin{lemma}\label{lem:nontrivial-bound}
The number of subproblems $I$ for which $\Qq_I$ is non-empty is at most $81\cdot (1/\delta\eps)^{4/\eps}\cdot |\Qq|p^2$.
\end{lemma}
\begin{proof}
It suffices to prove that for every fixed rectangle $Q_i\in \Qq$ and for every level $(s,t)\in \{1,\ldots,p\}\times\{1,\ldots,p\}$, there are at most $81\cdot(1/\delta\eps)^{4/\eps}$ subproblems $I$
of level $(s,t)$ for which $Q_i\in \Qq_I$. 
Observe that if $I=(s,t,x_1,x_2,y_1,y_2)$ and $I'=(s,t,x'_1,x'_2,y'_1,y'_2)$ are two subproblems of level $(s,t)$ such that $A_I\cap A_{I'}\neq \emptyset$, then we necessarily have 
\begin{eqnarray*}
x_1-(1/\delta\eps)^{1/\eps}<x'_1<x'_2<x_2+(1/\delta\eps)^{1/\eps} & \textrm{and}\\
y_1-(1/\delta\eps)^{1/\eps}<y'_1<y'_2<y_2+(1/\delta\eps)^{1/\eps}.& 
\end{eqnarray*}
Since $x_2\leq x_1+(1/\delta\eps)^{1/\eps}$ and $y_2\leq y_1+(1/\delta\eps)^{1/\eps}$, for fixed $I$ there are at most $81\cdot(1/\delta\eps)^{4/\eps}$ choices of $x'_1, x'_2,y'_1,y'_2$ that satisfy the above.
Therefore, for each subproblem $I$ of level $(s,t)$, there are at most $81\cdot(1/\delta\eps)^{4/\eps}$ other subproblems $I'$ of the same level for which $A_I\cap A_{I'}\neq \emptyset$.
If $Q_i\in \Qq_I$ for some $I$ of level $(s,t)$, then all subproblems $I'$ of the same level for which $Q_i\in \Qq_{I'}$ must satisfy $A_I\cap A_{I'}\neq \emptyset$, and hence there can be at
most $81\cdot(1/\delta\eps)^{4/\eps}$ of them.
\end{proof}

Having gathered all the tools, we conclude by describing the algorithm.
To compute $\OPT(\Qq)$, it is sufficient to compute $\Val(\glob)$ for the subproblem $\glob$ given by Lemma~\ref{lem:global-subproblem}.
As mentioned before, for this we use backtracking with memoization. 
Namely, starting from $\glob$, we recursively solve subproblems as explained in Lemma~\ref{lem:recursion}. 
Whenever $\Val(I)$ has been computed for some subproblem $I$, then this value is memoized in a map, and 
further calls to solving $I$ will only return the value retrieved from the map, instead of recomputing the value again.
Furthermore, whenever we attempt to compute $\Val(I)$ for a subproblem $I$ for which $\Qq_I$ is empty, we immediately return $0$ instead of applying the procedure of Lemma~\ref{lem:recursion}.
Therefore, the total running time of the algorithm is upper bounded by the number of subproblems $I$ for which $\Qq_I$ is non-empty, times the time spent on internal computations for each of them,
including checking whether 
the respective set $\Qq_I$ is empty and whether $\Val(I)$ has already been memoized.
The first factor is bounded by $81\cdot (1/\delta\eps)^{4/\eps}\cdot |\Qq|p^2\leq 81\cdot (1/\delta\eps)^{4/\eps}\cdot nN^2$ due to Lemma~\ref{lem:nontrivial-bound}.
The second factor is bounded by $f(\delta,\eps)\cdot n^d$ for some constant $d$, due to Lemma~\ref{lem:recursion}.
Hence, the running time of the whole algorithm is $f(\delta,\eps)\cdot (nN)^c$ for some computable function $f$ and constant $c$.
As mentioned before, the algorithm can be trivially adjusted to also compute an independent set 
of weight $\Val(\glob)$
by storing the value of each subproblem together with some independent set that certifies this value.

Summarizing, the dynamic programming described above computes an independent set in $\Qq$ of weight $\OPT(\Qq)$ in time $f(\delta,\eps)\cdot (nN)^c$.
As argued in the previous section, such an independent set projects to an independent set of the same weight in $\Rr^{-\delta}$, and $\OPT(\Qq)\geq (1-40\eps)\OPT$ holds with probability at least $1/2$.
This concludes the proof of Theorem~\ref{thm:main-approx}.

\subsection{Parameterized setting: Theorem~\ref{thm:main-fpt}}\label{sec:param}

In this section we present the proof of Theorem~\ref{thm:main-fpt}.
Note that if all rectangles had equal weights, then we could derive Theorem~\ref{thm:main-fpt} from Theorem~\ref{thm:main-approx} by simply running the algorithm for $\eps=1/(k+1)$.
This is because in the uniform-weight setting, a solution of size at most $k$ that loses at most a multiplicative factor of $1/(k+1)$ to the optimum, actually needs to be optimal.
However, this argument breaks when the weights are non-uniform. Instead, we argue how the reasoning from the proof of Theorem~\ref{thm:main-approx} should be adjusted 
to also justify Theorem~\ref{thm:main-fpt}.

First, we consider the sparsification arguments presented in Section~\ref{sec:sparse}, which yield subfamilies $\Rr'$ and $\Rr''$.
Let us apply them for $\eps=1/(20k)$.
Suppose $\Ss\subseteq \Rr$ is an optimum solution for the considered variant of the problem: it is an independent set of size at most $k$ and of total weight $\OPT_k(\Rr)$.
When applying the sparsification yielding family $\Rr'$ (Lemma~\ref{lem:well-separated}), every rectangle of $\Rr$ is removed with probability at most $2\eps=1/(10k)$ (cf. the proof of Lemma~\ref{lem:well-separated}).
Therefore, the expected number of rectangles from $\Ss$ that are removed when constructing $\Rr'$ is at most $1/10$.
Hence, by Markov's inequality, with probability at least $9/10$ none of them is removed, and $\Ss$ entirely survives in $\Rr'$.
Similarly, when $\Rr''$ is constructed from $\Rr'$, every rectangles is again removed with probability at most $8\eps=2/(5k)$ (Lemma~\ref{lem:abusive}).
Again, the same calculation using Markov's inequality yields that, conditioned on $\Ss$ surviving in $\Rr'$, $\Ss$ survives also in $\Rr''$ with probability at least $3/5$.
In total, we conclude that $\Ss$ survives in $\Rr''$ with probability at least $9/10\cdot 3/5 > 1/2$.
Note that if $\Ss\subseteq \Rr''$, then $\OPT_k(\Rr'')\geq \weight(\Ss)=\OPT_k(\Rr)$, but also $\OPT_k(\Rr'')\leq \OPT_k(\Rr)$ due to $\Rr''\subseteq \Rr$.
Consequently, in this case we have $\OPT_k(\Rr'')=\OPT_k(\Rr)$.

The construction of $\Qq$ from $\Rr''$ is kept unchanged.
Finally, in the dynamic programming we would like to compute $\OPT_k(\Qq)$ instead of $\OPT(\Qq)$.
The adjustment is standard: we add one additional coordinate to the definition of a subproblem, which keeps track of the number of rectangles to be picked in the solution.
More precisely, a subproblems is now a tuple $I=(s,t,x_1,x_2,y_1,y_2,\lambda)$, where the first six coordinates have the original meaning, while $\lambda$ is an integer with $0\leq \lambda\leq k$
that denotes the cardinality of the sought solution; $\lambda$ will be called the {\em{budget}}. 
In subproblem $I$ we are asked to compute $\Val(I)$, defined to be the largest weight of an independent set of size at most $\lambda$ contained in $\Qq_I$.
The recursive formula for computing $\Val(I)$, presented in Lemma~\ref{lem:recursion}, needs to be adjusted as follows: for each of the considered partitions into boxes, we also consider
all possible partitions of the budget among the boxes. Each large box has to receive budget $1$, since there will be exactly one rectangle fit into this box, and the budget $\lambda_B$ allocated  
to any other box $B$ becomes the budget in the corresponding subproblem $I_B$ that is solved recursively.
Note that since the number of boxes is bounded by $(1/\delta\eps)^{4/\eps}$ where $\eps=1/(20k)$, we have that the number of partitions of the budget $\lambda$ among the boxes is bounded by a computable
function of $\delta$ and $k$.

Wrapping up, with probability at least $1/2$ we have that $\OPT_k(\Rr'')=\OPT_k(\Rr)$.
Since $\Qq$ is obtained from $\Rr''$ only by shrinking the rectangles, we have $\OPT_k(\Qq)\geq \OPT_k(\Rr'')$.
As before, any independent set of rectangles in $\Qq$ projects to an independent set of the same weight in $\Rr^{-\delta}$.
Hence, if indeed $\OPT_k(\Rr'')=\OPT_k(\Rr)$, then the dynamic programming will output a set of rectangles of size at most $k$ that (after projection) is independent in $\Rr^{-\delta}$ and has weight 
at least $\OPT_k(\Rr)$.

\section{Kernelization results}\label{sec:kern}

In this section we present our kernelization results, which apply to the case when $\Rr$ consists of squares. As mentioned in the introduction, we can obtain a polynomial kernel for three settings: unit squares of uniform weight, unit squares of non-uniform weight, and squares of non-uniform size and uniform weight. 

We start our description by clarifying the definition of kernelization. Then, we present two auxiliary lemmas on (weighted) graphs. Later, we apply these lemmas to reduced instances of \mwisr on (unit) squares to obtain the different kernels. Finally, we explain how to reduce the bit encodings of the kernels, and 
we present applications for the design of EPTASes and fast FPT algorithms.


\subsection{Definition of kernel}
The classic definition of kernelization (cf.~\cite{CyganFKLMPPS15}) is tailored to decision problems; extending it to optimization problems in a weighted setting often turns out to be problematic.
Making the definition compatible with the $\delta$-shrinking model would make it
even more complicated.
For this reason, we define explicitly what we mean by kernelization for \mwisr in the shrinking model, bearing in mind the main principle of kernelization: solving the kernel should project
to a solution for the original instance.

\begin{definition}\label{def:kernelization}
Suppose $\Rr$ is a family of axis-parallel rectangles, $k$ is an integer parameter, and $\delta\in (0,1)$ is a constant.
Then, a {\em{kernel for $\Rr$ and $\delta$}} is a polynomial-time computable subfamily $\Qq\subseteq \Rr$ such that $|\Qq|\leq f(k,\delta)$ for a computable function~$f$, called the {\em{size}} of the kernel,
and the following holds:
$$\OPT_k(\Qq^{-\delta})\geq \OPT_k(\Rr).$$
\end{definition}

Thus, solving \mwisr on $\Rr$ with $\delta$-shrinking relaxation can be done by solving \mwisr on $\Qq^{-\delta}$ (without $\delta$-shrinking relaxation).
If one wishes to use an algorithm for the shrinking relaxation on the kernel, then applying it to $\Qq^{-\delta}$ with parameter $\delta$ will yield a subfamily $\Ss$ of size $k$ such that
$\Ss^{-2\delta}$ is independent, and $\weight(\Ss)\geq \OPT_k(\Qq^{-\delta})\geq \OPT_k(\Rr)$. Hence, this solves the original problem for $2\delta$-shrinking, and we can rescale the parameter accordingly.

A reader well-versed in the foundations of parameterized complexity will notice that Definition~\ref{def:kernelization} is at least as strong as any reasonable complexity-theoretical definition of kernelization, apart from one aspect. Namely, the weights and the coordinates of the rectangles in the kernel are inherited from the original instance, and hence, their bit encoding may not be bounded in terms of $k$ and~$\delta$. We prefer to work with Definition~\ref{def:kernelization}, because it focusses our efforts on the core combinatorial aspects of our kernelization procedures, rather than on the less relevant complexity-theoretical aspects. However, at the end of this section we argue that, essentially, the bit encodings of both the weights and the coordinates can be reduced to a polynomial in $k$ and $1/\delta$. 

\subsection{Auxiliary results on graphs}
We consider the parameterized \MWIS (\mwis) problem on graphs. In this problem, we are given a graph $G$, a non-negative weight function $w\colon V(G)\to \mathbb{R}_{\geq 0}$, and an integer parameter $k$; we look for a subset of $V(G)$ that has size at most $k$, contains only pairwise non-adjacent vertices (i.e.,~it is an \emph{independent set}), and has the maximum possible weight subject to these conditions. This maximum weight will be denoted by $\OPT_k(G)$. We present two simple reduction techniques that apply to parameterized \mwis on general graphs. 

We need the following notation. Given a graph $G$, let $\Ni{G}(v)$ denote the set of vertices at distance at most~$i$ from $v \in V(G)$, and let $\mdi(G)$ denote the maximum number of vertices at distance at most~$i$ from any vertex of $G$ (i.e.,~$\mdi(G) = \max_{v \in V(G)} |\Ni{G}(v)|)$. Note that $\md(G)$ is simply the maximum degree of $G$ plus $1$.

In the next lemma, we will identify a subgraph $G'$ of $G$ such that it suffices to search in $G'$ for the optimal solution to parameterized \mwis. Intuitively, later we will show that in the shrinking model for uniform squares 
the graph $G'$ can be assumed to have size 
$64k/\delta^2$ (weighted case), since the shrinking will allow us to ensure that $\mdt(G) \le 64/\delta^2$, for some appropriately chosen graph $G$.

\begin{lemma} \label{lem:kernel:Gweight}
Given an instance $(G,w,k)$ of parameterized \mwis, one can in polynomial time find an induced subgraph $G'$ of $G$ such that $\OPT_k(G) = \OPT_k(G')$ and $|V(G')| \leq k \cdot \mdt(G)$.
\end{lemma}
\begin{proof}
If $G$ is empty, then the lemma holds with $G'=G$, and thus we may assume that $G$ is not empty.
We iteratively define a sequence of graphs $G_0,G_1,\ldots,G_k$ and a sequence of vertices $v_1,\ldots,v_k \in V(G)$ as follows. First, we put $G_0=G$.
Consider $i=1,\ldots,k$ in order. If $G_{i-1}$ is not empty, then let $v_i$ be a vertex of maximum weight in $G_{i-1}$, and let $G_i = G_{i-1} - \Nt{G_{i-1}}[v_i]$. If $G_{i-1}$ is empty, then let $v_i = v_{i-1}$ and $G_i = G_{i-1}$.

We claim there is a set $I \subseteq V(G)$, such that $I$ is an independent set, $|I| \leq k$, $I$ has maximum weight subject to these conditions, and 
$I \cap V(G_k)=\emptyset$. If $G_k$ is an empty graph, then this is trivial. So assume that $G_k$ is not empty. Then, $v_1,\ldots,v_k$ are all distinct by construction.
Let $I' \subseteq V(G)$ be such that $I'$ is an independent set, $|I'| \leq k$, $I'$ has maximum weight under these conditions, and among all such sets $|I' \cap V(G_k)|$ is minimum. Suppose, for sake of contradiction, that $|I' \cap V(G_k)| > 0$. Let $x$ be an arbitrary vertex in $I' \cap V(G_k)$. By construction, the sets $\N{G}[v_1],\ldots,\N{G}[v_k]$ are pairwise disjoint, because if $\N{G}[v_j] \cap \N{G}[v_i] \not= \emptyset$ for some $j > i$, then $v_j \in \Nt{G_{i-1}}[v_i]$, which combined with the fact that $v_j \not= v_i$ contradicts the choice of $v_j$. Moreover, $V(G_k)$ is disjoint from the sets $\N{G}[v_1],\ldots,\N{G}[v_k]$ as well. As $|I'|\leq k$ and $x\in I' \cap V(G_k)$, it follows that $I'$ intersects at most $k-1$ sets among $\N{G}[v_1],\ldots,\N{G}[v_k]$. Consider any $i \in \{1,\ldots,k\}$ such that $I' \cap \N{G}[v_i] = \emptyset$. By the choice of $v_1,\ldots,v_k$, we have $w(v_i) \geq w(x)$. Therefore, the set $I'' = (I \setminus\{x\}) \cup \{v_i\}$ is an independent set, and it satisfies $|I''
| = |I'|$, $w(I'') \geq w(I')$, and $|I'' \cap V(G_k)| < |I' \cap V(G_k)|$; this contradicts the choice of $I'$. The claim follows.

Following the claim, $\OPT_k(G) = \OPT_k(G')$. Moreover, $|V(G')| \leq k \cdot \mdt(G)$ by construction. Finally, it is straightforward to find $G'$ in polynomial time by following the construction above.
\end{proof}
In the case of unit weights, a stronger lemma can be proven. In the shrinking model, it will imply later that either we can directly find an independent set of size $k$, or that due to the shrinking the number of input squares can be assumed to be at most $16k/\delta^2$. This is because the shrinking will allow us to argue that $\md(G)\le 16/\delta^2$, for some appropriately chosen graph $G$.

\begin{lemma} \label{lem:kernel:G}
Given an instance $(G,w,k)$ of parameterized \mwis with $w(v) = 1$ for all $v \in V(G)$ and $|V(G)|> k \cdot \md(G)$, one can in polynomial time find an independent set in $G$ of size $k = \OPT_k(G)$.
\end{lemma}
\begin{proof}
Again, we iteratively define a sequence of graphs $G=G_0,G_1,\ldots,G_k$ and a sequence of vertices $v_1,\ldots,v_k \in V(G)$ as follows. Consider $i=1,\ldots,k$ in order. To define $G_i$ based on $G_{i-1}$, let $v_i$ be any vertex in $G_{i-1}$, and let $G_i = G_{i-1} - \N{G_{i-1}}[v_i]$. As we will show next, $G_i$ never becomes empty, and thus this sequence is well-defined.

Note that each $G_i$ is an induced subgraph of $G$, and thus $\N{G_{i-1}}[v_i]$ is always of size at most $\md(G)$. By trivial induction, we have $|V(G_i)|\geq |V(G)|-i\cdot \md(G)$, which is always
positive due to $|V(G)|>k\cdot \md(G)$. Hence, the construction can be performed for all the $k$ steps, yielding distinct vertices $v_1,\ldots,v_k$.

Now observe that $\{v_1,\ldots,v_k\}$ is an independent set in $G$, because if $v_j$ and $v_i$ were adjacent for some $j > i$, then $v_j \in \N{G_{i-1}}[v_i]$, which combined with the fact that $v_j \not= v_i$ contradicts the choice of $v_j$. Hence, $\{v_1,\ldots,v_k\}$ can be returned. It is straightforward to perform the construction in polynomial time.
\end{proof}

\subsection{Kernels on unit squares}
To streamline the description of the kernels, we need the following definitions. The \emph{center} of a rectangle $R_i$ is the point $c_i =(x_i^{(1)} + \frac{g_i}{2},y_i^{(1)} + \frac{h_i}{2})$. For some $\lambda \geq 0$, we say that a family $\Rr$ of rectangles is \emph{$\lambda$-distant} if $||c_i-c_j||_\infty \geq \lambda$ for any distinct $R_i,R_j \in \Rr$. Also, if $\Rr$ is a family of squares and $\Qq\subseteq \Rr$ is a subfamily of $\Rr$, then we say that a mapping $\clmap \colon \Rr\to \Qq$ is {\em{$\lambda$-close}} if for every $R_i\in \Rr$ with $R_j=\clmap(R_i)$, we have $g_j\leq g_i$, $w_j\geq w_i$, and $||c_i-c_j||_\infty\leq \lambda$. Given such a mapping and a subfamily $\Ss$ of $\Rr$, we use $(\clmap(\Ss))^{-\delta}$ to denote the $\delta$-shrinking of the image of $\clmap$ on $\Ss$.

The next lemma will prove very useful.

\begin{lemma}\label{lem:close}
Let $\delta\in (0,1)$ and let $\Rr$ be a family of squares, each of side length at least $M$. Suppose further that $\clmap\colon \Rr \to \Qq$ is a $(\delta M/2)$-close mapping, for some subfamily $\Qq\subseteq \Rr$.
Then for every $R_i\in \Rr$, we have $R_j^{-\delta}\subseteq R_i$ where $R_j=\clmap(R_i)$.
\end{lemma}
\begin{proof}
Observe that 
$$R_i=\{p\,\colon\, ||p-c_i||_\infty< g_i/2\}\quad\textrm{and}\quad R_j^{-\delta}=\{p\,\colon\, ||p-c_j||_\infty< (1/2-\delta/2)\cdot g_j\}.$$
By the $(\delta M/2)$-closeness of $\clmap$, we have $g_j\leq g_i$ and $||c_i-c_j||\leq (\delta/2)\cdot M$.
Hence, for every $p\in R_j^{-\delta}$, we have
\begin{eqnarray*}
||p-c_i||_\infty & \leq & ||p-c_j||_\infty + ||c_i-c_j||_\infty\\
& < & (1/2-\delta/2)\cdot g_j+(\delta/2)\cdot M\leq (1/2-\delta/2)\cdot g_i+(\delta/2)\cdot g_i=g_i/2.
\end{eqnarray*}
Since $p$ was chosen arbitrarily, we conclude that $R_j^{-\delta}\subseteq R_i$.
\end{proof}

The following corollary is now almost immediate.

\begin{corollary}\label{cor:mapping}
Let $\delta\in (0,1)$, let $\Rr$ is a family of squares, each of side length at least $M$, let $\Qq\subseteq \Rr$ be a subfamily of $\Rr$, let $\clmap\colon \Rr \to \Qq$ be a $(\delta M/2)$-close mapping, 
and let $\Ss\subseteq \Rr$ be an independent subfamily of squares.
Then $\clmap$ is injective on $\Ss$, and $(\clmap(\Ss))^{-\delta}$ is an independent set in $\Qq^{-\delta}$ of weight not smaller than $w(\Ss)$.
Consequently, $\OPT_k(\Qq^{-\delta})\geq \OPT_k(\Rr)$ for each non-negative integer~$k$.
\end{corollary}
\begin{proof}
Observe that if $R_j=\clmap(R_i)=\clmap(R_{i'})$ for some $R_j,R_i,R_{i'}\in \Rr$, then by Lemma~\ref{lem:close} we have $R_j^{-\delta}\subseteq R_i\cap R_{i'}$, so in particular $R_i$ and $R_{i'}$ overlap.
Hence, $\clmap$ must be injective on any independent subfamily $\Ss$.
Moreover, Lemma~\ref{lem:close} ensures that $(\clmap(R_i))^{-\delta}\subseteq R_i$ for each $R_i\in \Ss$, so independence of $\Ss$ implies the independence of $(\clmap(\Ss))^{-\delta}$.
By the definition of a close mapping, each square is mapped to a square of at least the same weight, and hence $w(\clmap(\Ss))\geq w(\Ss)$.
The last assertion, that $\OPT_k(\Qq^{-\delta})\geq \OPT_k(\Rr)$, follows immediately from the previous ones and the definition of $\OPT_k(\cdot)$.
\end{proof}

We now give our main tool for the construction of kernels for unit squares; it can be seen as a sparsification of $\Rr$.

\begin{lemma} \label{lem:kernel:unit:tau}
Let $k$ be a non-negative integer, let $\delta \in (0,1)$, and let $\Rr$ be a family of unit squares of non-uniform weight. 
Then one can in polynomial time construct a $(\delta/2)$-distant subfamily $\Qq\subseteq \Rr$ together with a $(\delta/2)$-close mapping $\clmap\colon \Rr\to \Qq$.
\end{lemma}
\begin{proof}
To construct $\Qq$ and $\clmap$, we iteratively perform the following procedure. Take a square $R_j \in \Rr$ of maximum weight. Add $R_j$ to $\Qq$. For any square $R_i \in \Rr$ for which $||c_i - c_j||_\infty < (\delta/2)$, remove $R_i$ from $\Rr$ and set $\clmap(R_i) = R_j$. Iterate this procedure until $\Rr$ becomes empty. It is clear from the construction that $\Qq$ constructed in this manner is $(\delta/2)$-distant, and $\clmap\colon \Rr\to \Qq$ is $(\delta/2)$-close; this is by the choice of $R_j$ and the fact that all side lengths are equal. Moreover, the construction can be performed in polynomial time.
\end{proof}

We are ready to give a kernel for unit squares of uniform weight. In a nutshell, we use Lemma~\ref{lem:kernel:unit:tau} to find a $(\delta/2)$-distant subfamily $\Qq$ of $\Rr$. Then we transfer an argument made by Alber and Fiala~\cite[Corollary~3.1]{AlberF04} to kernelize $\lambda$-distant families of unit disks to the setting of unit squares, and apply it to $\Qq$ to obtain the final kernel.

\begin{theorem} \label{thm:kernel:unit:unit}
Let $k$ be a non-negative integer, let $\delta \in (0,1)$, and let $\Rr$ be a family of unit squares of uniform weight. Then there is a kernel for $\Rr$ and $\delta$ of size at most $16k/\delta^2$.
\end{theorem}
\begin{proof}
Apply the algorithm of Lemma~\ref{lem:kernel:unit:tau} to $\Rr$, thereby constructing a $(\delta/2)$-distant subfamily $\Qq \subseteq \Rr$ and a $(\delta/2)$-close mapping $\clmap\colon \Rr\to \Qq$.
By Corollary~\ref{cor:mapping}, we have $\OPT_k(\Qq^{-\delta})\geq \OPT_k(\Rr)$.
Consider the intersection graph $G$ of $\Qq^{-\delta}$.

First, suppose that $|\Qq^{-\delta}|=|V(G)|>k \cdot \md(G)$. Then, we can apply Lemma~\ref{lem:kernel:G} to $G$, thus in polynomial time finding an independent set $I$ in $G$ of size $k=\OPT_k(G)$.
Then, we return the subfamily $\Ss$ of squares in $\Rr$ that corresponds to $I$. Note that $k = w(\Ss) = \OPT_k(\Qq^{-\delta}) \geq \OPT_k(\Rr)$ and that $\Ss^{-\delta}$ is an independent set. Hence, $\Ss$ satisfies the conditions of a kernel for $\Rr$ and $\delta$ of size at most $k$.

Second, suppose that $|\Qq^{-\delta}|\leq |V(G)|\leq k \cdot \md(G)$. If we can bound $\md(G)$ by $(4/\delta)^2$, then $\Qq$ satisfies the conditions of a kernel for $\Rr$ and $\delta$ of size at most $16k/\delta^2$, because $\Qq$ is a subfamily of $\Rr$ and $\OPT_k(\Qq^{-\delta}) \geq \OPT_k(\Rr)$.

In order to bound $\md(G)$ by $(4/\delta)^2$, we prove that any square $R^{-\delta}_i \in \Qq^{-\delta}$ overlaps at most $(4/\delta)^2$ squares of $\Qq^{-\delta}$ (including itself). This implies the same bound on $\md(G)$. Observe that any square $R^{-\delta}_j \in \Qq^{-\delta}$ that overlaps $R^{-\delta}_i$ has its center inside a square of side length $2-2\delta$ centered on $c_i$. Since $\Qq$ is $(\delta/2)$-distant, so is $\Qq^{-\delta}$, and thus
open squares with side length $(\delta/2)$ centered at the centers of squares of $\Qq^{-\delta}$ are pairwise disjoint.
The union of all these open squares is contained in a square of side length $2-\delta$ centered on $c_i$.
Then a simple area bound shows that at most $(4/\delta)^2$ squares of $\Qq^{-\delta}$ overlap $R_i$. The theorem follows.
\end{proof}

\begin{theorem} \label{thm:kernel:unit}
Let $k$ be a non-negative integer, let $\delta \in (0,1)$, and let $\Rr$ be a family of unit squares of non-uniform weight. Then there is a kernel for $\Rr$ and $\delta$ of size at most $64k/\delta^2$.
\end{theorem}
\begin{proof}
Again, apply the algorithm of Lemma~\ref{lem:kernel:unit:tau} to $\Rr$, thereby constructing a $(\delta/2)$-distant subfamily $\Qq$ and a $(\delta/2)$-close mapping $\clmap\colon \Rr\to \Qq$; in particular $\OPT_k(\Qq^{-\delta})\geq \OPT_k(\Rr)$ by Corollary~\ref{cor:mapping}.
Let $G$ be the intersection graph of $\Qq^{-\delta}$. We proceed identically as in the proof of Theorem~\ref{thm:kernel:unit:unit}, except that we replace our use of Lemma~\ref{lem:kernel:G} by Lemma~\ref{lem:kernel:Gweight}. 
More precisely, by applying the algorithm of Lemma~\ref{lem:kernel:Gweight} on $G$, we compute an induced subgraph $G'$ such that $\OPT_k(G')=\OPT_k(G)=\OPT_k(\Qq^{-\delta})$ and $|V(G')|\leq k\cdot \mdt(G)$.
This induced subgraph corresponds to a family of squares that forms a kernel. Hence, it suffices to bound $\mdt(G)$ by $64/\delta^2$. We prove that for any square $R^{-\delta}_i \in \Qq^{-\delta}$, there are at most $64/\delta^2$ squares $R^{-\delta}_j \in \Qq^{-\delta}$ that either overlap $R_i$, or that overlap a square of $\Qq^{-\delta}$ that overlaps $R^{-\delta}_i$. Indeed, such $R^{-\delta}_j$ must have its center inside a square of side length $4-4\delta$ centered on $c_i$. Since $\Qq$ is $(\delta/2)$-distant, a similar area bound as in Theorem~\ref{thm:kernel:unit:unit} shows that at most $64/\delta^2$ squares of 
$\Qq^{-\delta}$ overlap $R^{-\delta}_i$.
\end{proof}

\subsection{Kernel for squares of non-uniform size and uniform weight}
We first require the following auxiliary lemma, which superficially resembles Lemma~\ref{lem:kernel:unit:tau}.
There is, however, a key difference: Lemma~\ref{lem:kernel:unit:tau} can be applied to unit squares of non-uniform weights, while the next result treats squares of uniform weight and non-uniform size.
This is reflected by the goal of the greedy choice in the proof.

\begin{lemma} \label{lem:kernel:arb:tau}
Let $k$ be a non-negative integer, let $\delta \in (0,1)$, and let $\Rr$ be a family of squares of uniform weight, each of side length at least $M$ for some positive real $M$.
Then one can in polynomial time construct a $(\delta M/2)$-distant subfamily $\Qq\subseteq \Rr$ together with a $(\delta M/2)$-close mapping $\clmap\colon \Rr\to \Qq$.
\end{lemma}
\begin{proof}
The proof is similar to that of Lemma~\ref{lem:kernel:unit:tau}, except that we use the side length of a square in order to create a $(\delta M/2)$-distant subfamily, rather than the weight.
To construct $\Qq$ and $\clmap$, we iteratively perform the following procedure. Take a square $R_j \in \Rr$ of minimum side length. Add $R_j$ to $\Qq$. For any square $R_i \in \Rr$ for which $||c_i - c_j||_\infty < (\delta M/2)$, remove $R_i$ from $\Rr$ and set $\clmap(R_i) = R_j$. Iterate this procedure until $\Rr$ becomes empty. Again, it is clear from the construction that $\Qq$ constructed in this manner is $(\delta M/2)$-distant and $\clmap\colon \Rr\to \Qq$ is $(\delta M/2)$-close; this is by the choice of $R_j$ and the fact that all weight are equal. Moreover, the construction can be performed in polynomial time.
\end{proof}
This lemma by itself leads to the following generalization of Theorem~\ref{thm:kernel:unit:unit}, which
yields a kernel for instances in which the sizes of the input squares are in a bounded range. In a nutshell, we use Lemma~\ref{lem:kernel:arb:tau} and then transfer an argument made by Alber and Fiala~\cite[Corollary~3.1]{AlberF04} to kernelize $\lambda$-distant families of disks of bounded radius to our setting, just like Theorem~\ref{thm:kernel:unit:unit}.

\begin{theorem} \label{thm:kernel:arb:bounded}
Let $k$ be a non-negative integer, let $\delta \in (0,1)$, let $M_1,M_2$ be real numbers such that $1 \leq M_1 \leq M_2$, and let $\Rr$ be a family of squares of uniform weight, each of side length between $M_1$ and $M_2$. Then there is a kernel for $\Rr$ and $\delta$ of size at most $16k/\delta^2\cdot (M_2/M_1)^2$.
\end{theorem}
\begin{proof}
Apply Lemma~\ref{lem:kernel:arb:tau} to $\Rr$ and $M=M_1$, thereby obtaining a $(\delta M_1/2)$-distant subfamily $\Qq$ and a $(\delta M_1/2)$-close mapping $\clmap\colon \Rr\to \Qq$; again, 
$\OPT_k(\Qq^{-\delta})\geq \OPT_k(\Rr)$ by Corollary~\ref{cor:mapping}. Let $G$ be the intersection graph of $\Qq^{-\delta}$. We proceed identically as in the proof of Theorem~\ref{thm:kernel:unit:unit}, except for modifying the argumentation leading to the upper bound on $\md(G)$. Namely, observe that any square $R^{-\delta}_j \in \Qq^{-\delta}$ that overlaps $R^{-\delta}_i$ has its center inside a square of side length $2M_2-2\delta M_1$ centered on $c_i$. Since $\Qq$ is $M_1 \cdot (\delta/2)$-distant, a similar area bound as in Theorem~\ref{thm:kernel:unit:unit} shows that at most $(4M_2/(M_1\delta))^2$ squares of $\Qq$ overlap $R_i$.
\end{proof}
In order to extend this result to squares of arbitrary size, we need the following auxiliary lemma, which encapsulates the core combinatorial argument. It essentially states that if the sizes of two squares differ by a lot, then either one is contained in the other, or their shrunk counterparts are disjoint.

\newcommand{\bnd}{\partial}

\begin{lemma} \label{lem:kernel:arb:diff}
Let $\delta \in (0,1)$ and let $R_i, R_j$ be any two squares for which $g_i \leq g_j$. Then $R_i$ is contained in $R_j$, or $R_i^{-\delta}$ and $R_j^{-\delta}$ do not overlap, or $\frac{g_j}{g_i} < 2/\delta$.
\end{lemma}
\begin{proof}
Suppose that $R_i$ is not contained in $R_j$, and $R_i^{-\delta}$ and $R_j^{-\delta}$ overlap. We are going to prove that $\frac{g_j}{g_i} < 2/\delta$.
Let $\bnd R_i$ and $\bnd R_j$ denote the boundaries of $R_i$ and $R_j$ respectively.
Observe that $R_i$ and $R_j$ must overlap, since otherwise $R_i^{-\delta}$ and $R_j^{-\delta}$ would be disjoint.
As $R_i$ has not larger side length than $R_j$, and $R_i$ is not contained in $R_j$, the boundaries $\bnd R_i$ and $\bnd R_j$ must intersect.
Let $x$ be any point in $\bnd R_i\cap \bnd R_j$.

Observe now that if $R_i^{-\delta}$ was contained in $R_j^{-\delta}$, then $R_i$ would be contained in $R_j$, because $g_i\leq g_j$; this is contradicts the assumptions.
Since $R_i^{-\delta}$ and $R_j^{-\delta}$ overlap, we conclude that their boundaries $\bnd R_i^{-\delta}$ and $\bnd R_j^{-\delta}$ also intersect.
Let $y$ be any point in $\bnd R_i^{-\delta}\cap \bnd R_j^{-\delta}$.

We now examine $||x-y||_\infty$, the $\ell_\infty$-distance between $x$ and $y$.
On the one hand, we have $x\in \bnd R_i$ and $y\in \bnd R^{-\delta}_i\subseteq R_i$. Hence,
\begin{equation}\label{eq:upper}
||x-y||_\infty\leq g_i - \frac{\delta}{2} \leq g_i,
\end{equation}
by noting that the $\ell_\infty$-distance between any point of $\bnd R_i$ and any point of $\bnd R^{-\delta}_i$ is at most $g_i - \frac{\delta}{2} \leq g_i$.

On the other hand, $x\in \bnd R_j$ and $y\in \bnd R^{-\delta}_j$. Observe that each point of $\bnd R^{-\delta}_j$ is at $\ell_\infty$-distance at least $\delta/2\cdot g_j$ from any point of $\bnd R_j$. Hence,
\begin{equation}\label{eq:lower}
\delta/2\cdot g_j\leq ||x-y||_\infty.
\end{equation}
By combining (\ref{eq:upper}) and (\ref{eq:lower}), we conclude that $\frac{g_j}{g_i} < 2/\delta$.
\end{proof}
We can now give the kernel for squares of arbitrary size and uniform weight.

\begin{theorem}\label{thm:kernel:arb-sizes}
Let $k$ be a non-negative integer, let $\delta \in (0,1)$, and let $\Rr$ be a family of squares of uniform weight and non-uniform size. Then there is a kernel for $\Rr$ and $\delta$ of size 
$\Oh( k^2\cdot \frac{\log (1/\delta)}{\delta^3})$.
\end{theorem}
\begin{proof}
The kernelization algorithm consists of several stages. In the first stage, we reduce $\Rr$ following a simple rule. If $\Rr$ contains two squares $R_i, R_j$ such that $R_i$ is contained in $R_j$, then in any subfamily $\Ss$ of $\Rr$ that is independent we can replace $R_j$ by $R_i$; this replacement will not change the total weight, because the weights are uniform. Hence, removing $R_j$ from $\Rr$ yields a family with the same value of $\OPT_k(\cdot)$. We apply this rule exhaustively. 
The resulting family, which we also call $\Rr$ by abuse of notation, has the same value of $\OPT_k(\cdot)$ as the original one, and thus the rule is safe.
That is, from now on we may assume that for any two squares in $\Rr$, neither one square is contained in the other.

In the second stage, we partition the squares into levels according to their side length. For a non-negative integer $\ell$, let 
$$\Rr_{\ell} = \{R_i \in \Rr \mid (1+\delta)^\ell \leq g_i < (1+\delta)^{\ell+1}\}$$
be the family of squares of \emph{level} $\ell$. We call a level $\ell$ \emph{active} if $\Rr_\ell \not= \emptyset$.
Let $\gamma = 1 + \log_{1+\delta} (2/\delta)$. Thus, $\gamma$ is essentially the $\log_{1+\delta}$ of the constant of Lemma~\ref{lem:kernel:arb:diff}, and note that $\gamma = \Oh(\log_{1+\delta} (1/\delta))=\Oh(1/\delta\cdot \log 1/\delta)$.

\begin{claim} \label{claim:gamma}
If $\Rr$ defines at least $\gamma \cdot k$ active levels, then one can find in polynomial time a subfamily $\Pp$ of $\Rr$ of size $k$ such that $\Pp^{-\delta}$ is independent.
\end{claim}
\begin{proof}
By assumption, there are active levels $\ell_1,\ldots,\ell_k$ such that $\ell_{j} \geq \ell_i + \gamma$ for all $i,j\in\{1,\ldots,k\}$, $i<j$. Then 
$$(1+\delta)^{\ell_{j}} \geq (1+\delta)^{\ell_{i} + 1 + \log_{1+\delta} (2/\delta)} \geq (1+\delta)^{\ell_i+1} \cdot (2/\delta).$$
Then for any two squares $R_i,R_j$ in levels $\ell_i$ and $\ell_j$ respectively, where $i<j$, we have $\frac{g_j}{g_i} \geq 2/\delta$; this is by the definition of the levels. By Lemma~\ref{lem:kernel:arb:diff}, we infer that $R_i^{-\delta}$ and $R_j^{-\delta}$ have to be disjoint, because $R_i$ is not contained in $R_j$ by the assumption on $\Rr$ after the first stage.

For any $i \in\{1,\ldots,k\}$, let $R_i$ be an arbitrary square of level $\ell_i$; this is properly defined because $\ell_i$ is an active level. The preceding argument shows that if $\Pp = \{R_i \colon i \in \{1,\ldots,k\}\}$, then $\Pp^{-\delta}$ is an independent set of size and weight $k$. 
It now suffices to observe that $\Pp$ can be constructed in polynomial time by determining the level of each square, and counting the number of active levels. Since only the first $\log_{1+\delta} \max_{R_i \in \Rr} \{x_i^{(2)},y_i^{(2)}\}$ levels may be active, this is indeed polynomial in the size of the input.
\cqed\end{proof}

In the third stage, we consider two cases, depending on the number of active levels defined by $\Rr$.
Suppose first that $\Rr$ defines at least $\gamma \cdot k$ active levels. Apply Claim~\ref{claim:gamma}, thereby constructing a subfamily $\Pp\subseteq \Rr$. Then $\OPT_k(\Pp^{-\delta}) = k \geq \OPT_k(\Rr)$ by the properties of $\Pp$, following Claim~\ref{claim:gamma} and the fact that all squares have uniform weight. Hence, $\Pp$ satisfies the conditions for a kernel for $\Rr$ and $\delta$ of size $k$, and the algorithm can output $\Pp$.

Suppose then that $\Rr$ defines less than $\gamma \cdot k$ active levels. For each active level~$\ell$, we apply the kernel of Theorem~\ref{thm:kernel:arb:bounded} to $\Rr_\ell$ with parameters $k$, $M_1=(1+\delta)^\ell$, and $M_2=(1+\delta)^{\ell+1}$. 
By inspection of the proof of Theorem~\ref{thm:kernel:arb:bounded} (and Theorem~\ref{thm:kernel:unit:unit}), this will either yield an independent subfamily $\Ss_\ell$ of $\Rr_\ell$ of size~$k$,
or provide a $(\delta M_1/2)$-distant subfamily $\Qq_\ell\subseteq \Rr_\ell$ of size at most $16k(1+\delta)^2/\delta^2$, together with a $(\delta M_1/2)$-close mapping $f_\ell\colon \Rr_\ell\to \Qq_\ell$.
If the first case happens for some layer $\ell$, then $|\Ss_\ell|=w(\Ss_\ell)=k=\OPT_k(\Rr)$. Hence, $\Ss_\ell$ can be output as a kernel for $\Rr$ and $\delta$ of size $k$.
If the first case never happens, that is, the second case happens for all active levels $\ell$, then
define $\Qq$ to be the union of $\Qq_\ell$ through all active levels $\ell$. Then
$$|\Qq|\leq (\gamma\cdot k)\cdot 16k(1+\delta)^2/\delta^2=\Oh\left(k^2\cdot \frac{\log (1/\delta)}{\delta^3}\right).$$
Moreover, $\Qq$ can be computed in polynomial time. 

Now in order to argue that $\Qq$ is a kernel for $\Rr$ and $\delta$, it suffices to show that $\OPT_k(\Qq^{-\delta})\geq \OPT_k(\Rr)$.
Take an independent subfamily $\Ss\subseteq \Rr$ with $|\Ss|=w(\Ss)=\OPT_k(\Rr)$.
Let us define $\Ss'$ to be the union of $f_\ell(\Ss\cap \Rr_\ell)$ for all active levels $\ell$.
Clearly $\Ss'\subseteq \Qq$.
By Lemma~\ref{lem:close}, $f_\ell$ maps each square $R_i\in \Ss\cap \Rr_\ell$ to a square $R_j=f_\ell(R_i)$ such that $R_j^{-\delta}\subseteq R_i$.
Consequently, since $\Ss$ is independent, we infer that $(\Ss')^{-\delta}$ is also independent and $|\Ss'|=|\Ss|$.
This means that $\OPT_k(\Qq^{-\delta})\geq |\Ss'|=|\Ss|=\OPT_k(\Rr)$.
Hence, $\Qq$ indeed can be output as a kernel for $\Rr$ and $\delta$ of size $\Oh\left(k^2\cdot \frac{\log (1/\delta)}{\delta^3}\right)$.
\end{proof}

\subsection{Bit encodings} 
The algorithms of Theorems~\ref{thm:kernel:unit:unit},~\ref{thm:kernel:unit}, and~\ref{thm:kernel:arb-sizes} can be used to reduce the search space to a subfamily of size bounded in terms of $k$ and $1/\delta$, when solving the appropriate subcases \mwisr with $\delta$-shrinking. However, in the field of kernelization, one often measures the size of a kernel by the bit size of its encoding.
Definition~\ref{def:kernelization} forbids changing the weights or the coordinates of the rectangles in the kernel. Hence, if we follow the definition, then we cannot bound the bit size of the encoding as a function of $k$ and $1/\delta$. However, after computing the kernel, we can still discuss whether it is possible to replace the weights and the coordinates with numbers with a smaller encoding, so that the relevant combinatorial properties of the instance are preserved. We now show that this essentially can be done using known results, at least for the kernels obtained in this work.

First, we address the weights. Using the recent technique of Etscheid et al.~\cite{EtscheidKMR15}, we can assume that the weight function for a kernel $\Qq$ can be encoded using $\Oh(|\Qq|^4)$ bits. Indeed, they show in \cite[Example~1]{EtscheidKMR15} that for any family~$X$, weight function $w \colon X \to \mathbb{Q}$, and $W \in \mathbb{Q}$, one can in polynomial time find a weight function $w' \colon X \to \mathbb{Q}$ and $W' \in \mathbb{Q}$ of total encoding length $\Oh(|X|^4)$ such that $w(Y) \geq W$ if and only if $w'(Y) \geq W'$ for any $Y \subseteq X$. Hence, if we transform our parameterized setting to a parameterized decision setting, where $\OPT_k$ must be at least $W$, then we can find a suitable encoding of the weight function.
Our kernels do not rely on the a-priori knowledge of $W$, and therefore, we did not describe our kernels in the decision setting and remained more general instead.

Second, we address the representation of the squares. In M{\"u}ller et al.~\cite{MullervLvL2013}, it is shown that any intersection graph of axis-parallel squares with $n$ vertices has a representation using $\Oh(n \log n)$ or $\Oh(n^2)$ bits, depending on whether the squares have uniform or non-uniform size. However, actually finding a representation is NP-hard~\cite{Breu-thesis,Yannakakis}. 

Fortunately, the result of M{\"u}ller et al.~\cite[Theorems~3 and 4]{MullervLvL2013} implies that when a representation of an intersection graph of (unit) squares is already given, then a representation with $\Oh(n \log n)$ or $\Oh(n^2)$ bits respectively can be found in polynomial time. We briefly describe the intuition. In the case of unit squares (see~\cite[Theorem~3]{MullervLvL2013}), a representation can be decomposed into two unit interval graphs by projecting on the axes, and a new representation for both these graphs can be computed in polynomial time using $\Oh(n \log n)$ bits. In the case of arbitrary squares (see~\cite[Theorem~4]{MullervLvL2013}), a representation can be transformed into a linear program, and a new representation can be found by using Cramer's rule to solve the linear program. This also shows that the NP-hardness of finding a representation is not in giving the actual numbers, but in finding the right ``ordering'' of the squares.

The above results imply that if a kernel uses a subfamily of size $f(k)$ of the squares in the input, which is the case for our kernels, then a representation of the squares in the kernel using $\Oh(f(k) \log f(k))$ or $\Oh(f^2(k))$ bits can be found in polynomial time, because the original squares imply a representation of the subfamily. An important caveat, however, is that finding such a new representation might change what squares are disjoint after shrinking by $\delta$. Hence, if $\Qq$ is the kernel, then we should rather find a new representation of $\Qq^{-\delta}$. This makes sense if we think of the original motivation of the problem, which is to reduce the problem to solving \mwisr in $\Qq^{-\delta}$ (see also the remark below Definition~\ref{def:kernelization}).

\subsection{Applications}
We now present applications of our kernelization results for the design of FPT algorithms and EPTASes for the corresponding variants of \mwisr. Recall that our kernelization results only apply to squares, and therefore, the applications in this section are significantly less general than Theorem~\ref{thm:main-approx} and Theorem~\ref{thm:main-fpt}. However, the resulting algorithms are substantially faster.

Throughout this section, for brevity we use the $\Oh_\rho(\cdot)$-notation that, for a parameter $\rho$, hides factors that depend only on $\rho$.

\paragraph*{Parameterized algorithms.}
We obtain faster parameterized algorithms for squares of uniform weight or uniform size. In particular, we will obtain so-called \emph{subexponential} algorithms, where the exponent depends sublinearly on $k$. The design of such algorithms in the planar and geometric setting is an important direction in parameterized complexity research, see e.g.~\cite{DemaineFHT05,FominLMPPS16,KleinM14,MarxP15,PilipczukPSL13,PilipczukPSL14}. We note that the exponent of Theorem~\ref{thm:main-fpt}, which applies to arbitrary rectangles, is polynomial in $k$. Hence, Theorem~\ref{thm:app-FPT1} and~\ref{thm:app-FPT2} below are a major improvement whenever all the input rectangles are squares of uniform weight or of uniform size.

We present two different and incomparable results. The first is obtained by combining our kernelization with the algorithm of Marx and Pilipczuk~\cite{MarxP15}, and for the second we similarly use the results of Alber and Fiala~\cite{AlberF04}.

\begin{theorem}\label{thm:app-FPT1}
There is an algorithm that given a weighted family $\Rr$ of $n$ axis-parallel squares of total encoding size $N$ and of uniform size or of uniform weight, 
and parameters $k$ and $\delta$, runs in time $(k/\delta)^{\Oh(\sqrt{k})}\cdot (nN)^{c}$ for some  constant $c$, and outputs a subfamily $\Ss\subseteq \Rr$ such that $|\Ss|\leq k$, $\Ss^{-\delta}$ is independent, and $\weight(\Ss)\geq \OPT_k(\Rr)$.
\end{theorem}
\begin{proof}
We apply a result of Marx and Pilipczuk~\cite{MarxP15}. They gave an algorithm that, given a family $\mathcal{D}$ of $n$ polygons of non-uniform weight\footnote{Marx and Pilipczuk state this result explicitly for the unweighted setting (Theorem~1.5 in~\cite{MarxP15}), but the result for the weighted setting follows from the most general statement (Theorem~3.1 in~\cite{MarxP15}).} in the plane and a parameter~$k$, finds a subfamily of polygons that is independent and has size at most $k$,
and which achieves the maximum possible weight subject to these constraints. The running time of the algorithm is $|\mathcal{D}|^{\Oh(\sqrt{k})}\cdot (nN)^{\Oh(1)}$, where $N$ is the total bit size of the input.
By applying this algorithm on $\Qq^{-\delta}$, for a kernel $\Qq$ obtained by applying any of Theorem~\ref{thm:kernel:unit:unit},~\ref{thm:kernel:unit}, or~\ref{thm:kernel:arb-sizes},
we can compute $\OPT_k(\Qq^{-\delta})$ in time $|\Qq^{-\delta}|^{\Oh(\sqrt{k})}\cdot (nN)^{\Oh(1)}=(k/\delta)^{\Oh(\sqrt{k})}\cdot (nN)^{\Oh(1)}$. 
\end{proof}


\begin{theorem}\label{thm:app-FPT2}
There is an algorithm that given a weighted family $\Rr$ of $n$ axis-parallel squares of total encoding size $N$, uniform size, and non-uniform weight, 
and parameters $k$ and $\delta$, runs in time $2^{\Oh_\delta(\sqrt{k})}\cdot (nN)^{c}$ for some  constant $c$, and outputs a subfamily $\Ss\subseteq \Rr$ such that $|\Ss|\leq k$, $\Ss^{-\delta}$ is independent, and $\weight(\Ss)\geq \OPT_k(\Rr)$.
\end{theorem}
\begin{proof}
Let $\Qq$ be the family returned by Theorem~\ref{thm:kernel:unit} applied to $k$, $\delta$, and $\Rr$. Observe that $\Qq$ can be computed in polynomial time and that $|\Qq| = \Oh_\delta(k)$. Moreover, an inspection of the proof of Theorem~\ref{thm:kernel:unit} reveals that $\Qq$ is $(\delta/2)$-distant, and thus the same holds for $\Qq^{-\delta}$. A slight adaptation of a result of Alber and Fiala~\cite[Corollary~4.1]{AlberF04} (see also~\cite[Theorem~6.2.4]{proefschrift}) gives an algorithm that, given a family $\mathcal{D}$ of $n$ $\lambda$-distant unit squares of non-uniform weight\footnote{Alber and Fiala state this result explicitly for the unweighted setting~\cite[Corollary~4.1]{AlberF04}, but the result for the weighted setting follows in exactly the same way. Moreover, instead of their Theorem~4.2, we need to rely on a separator theorem by Smith and Wormald~\cite[Theorem~26]{SmithW98}. See also~\cite[Theorem~6.2.4]{proefschrift}.} in the plane, finds a subfamily of $\mathcal{D}$ that is independent and achieves the maximum possible weight subject to this constraint. The running time of the algorithm is $2^{\Oh_\lambda(\sqrt{|\mathcal{D}|})}\cdot (nN)^{\Oh(1)}$, where $N$ is the total bit size of the input. By applying this algorithm on $\Qq^{-\delta}$, we can compute $\OPT_k(\Qq^{-\delta})$ in time $2^{\Oh_\delta(\sqrt{k})}\cdot (nN)^{\Oh(1)}$.
\end{proof}
We can also obtain a faster algorithm for squares of uniform weight and non-uniform size by combining the results of Alber and Fiala~\cite{AlberF04} with the kernel of Theorem~\ref{thm:kernel:arb-sizes} and a later lemma (Lemma~\ref{lem:app-ply-arb}) about the ply of the objects. However, the running time obtained in this way is $2^{\Oh_{\delta}(k)}\cdot (nN)^{c}$, which is worse than what we can obtain through Theorem~\ref{thm:app-FPT1}.

\paragraph*{Approximation algorithms.}
Next, we show how to obtain faster approximation algorithms for squares of uniform weight or of uniform size. In order to prove these results, we need several auxiliary steps, and the following notion in particular. The \emph{ply} of a family of geometric objects in the two-dimensional plane is defined as the maximum number of objects in the family that contain any point in the plane. This is sometimes also called the \emph{maximum depth} of any point. This enables us to state following auxiliary lemma, which transfers a result of Alber and Fiala~\cite[Lemma~4.1]{AlberF04} from disks to squares.

\begin{lemma} \label{lem:app-ply}
Let $\lambda > 0$, let $M$ be a real number, and let $\Rr$ be a $\lambda$-distant family of squares, each of side length at most $M$. Then $\Rr$ has ply $\Oh(M^2 / \lambda^2)$.
\end{lemma}
\begin{proof}
This follows using essentially the same area bound as in Theorem~\ref{thm:kernel:arb:bounded} and~\ref{thm:kernel:unit:unit}; see also the work of Alber and Fiala~\cite[Lemma~4.1]{AlberF04}.
\end{proof}
Now we can prove the following lemma.


\begin{lemma} \label{lem:app-ply-arb}
Let $k$ be a non-negative integer, let $\delta \in (0,1)$, let $\Rr$ be a family of squares of uniform weight and non-uniform size, and let $\Qq \subseteq \Rr$ be the family returned by Theorem~\ref{thm:kernel:arb-sizes} when applied to $k$, $\delta$, and $\Rr$. Then the ply of $\Qq^{-\delta}$ is $\Oh_\delta(1)$.
\end{lemma}
\begin{proof}
If $\Qq^{-\delta}$ is an independent set, then $\Qq^{-\delta}$ has ply at most $1$. Hence, we can focus on the final case of the proof of Theorem~\ref{thm:kernel:arb-sizes}, i.e., when 
the kernel of Theorem~\ref{thm:kernel:arb:bounded} is applied to each active level separately, for this is the only case when the algorithm does not return an independent set of size $k$.
We adopt the notation from the proof of Theorem~\ref{thm:kernel:arb:bounded}.

Let $p$ be any point in the plane. We claim that any two squares of $\Qq^{-\delta}$ that contain~$p$ differ by $\Oh(\log_{1+\delta} (2/\delta))$ levels (recall that levels were defined w.r.t.\ the original squares from $\Qq$, but we consider their corresponding shrunken squares in $\Qq^{-\delta}$). Let $R_i^{-\delta} \in \Qq^{-\delta}$ be any square that contains $p$. We now consider any other square $R_j^{-\delta} \in \Qq^{-\delta}$ that contains~$p$. Assume w.l.o.g.\ that $g_i \leq g_j$ (the other case is symmetric). By the construction of $\Qq$, $R_i$ cannot contain $R_j$ or vice versa. Since $R_i^{-\delta}$ and $R_j^{-\delta}$ overlap (in $p$), it follows from Lemma~\ref{lem:kernel:arb:diff} that $\frac{g_j}{g_i} < 2/\delta$. Since levels are delimited by powers of $(1+\delta)$, the claim follows.

Now observe that each level of $\Qq^{-\delta}$ has ply $\Oh_\delta(1)$ by Lemma~\ref{lem:app-ply}, because by construction, each level $\ell$ of $\Qq^{-\delta}$ is a family of $(\delta(1+\delta)^\ell/2)$-distant squares, each of side length at most $(1+\delta)^{\ell+1}$. Combined with the claim, this implies that the ply of $\Qq^{-\delta}$ is $\Oh_\delta(1)$.
\end{proof}

The next theorem gives our EPTAS for axis-parallel squares. The existence of an EPTAS already follows from Theorem~\ref{thm:main-approx}, but we state it because the obtained dependence of the running time on $\eps$ is far better than that of Theorem~\ref{thm:main-approx}.

\begin{theorem}
There is an algorithm that given a weighted family $\Rr$ of $n$ axis-parallel squares with total encoding size $N$, of uniform size or of uniform weight, 
and parameters $\delta,\eps$, runs in time $2^{\Oh_\delta(1/\eps)} \cdot (nN)^{c}$ for some constant $c$, and outputs a subfamily $\Ss\subseteq \Rr$ such that $\Ss^{-\delta}$ is independent, and $\weight(\Ss)\geq (1-\eps)\OPT(\Rr)$.
\end{theorem}
\begin{proof}
For each $k = 1,\ldots,n$, we will construct a family $\Qq_k \subseteq \Rr_k$ and a family $\Ss_k \subseteq \Qq_k$ such that $\OPT_k(\Qq_k^{-\delta}) \geq \OPT_k(\Rr)$, $\Ss_k^{-\delta}$ is independent, and $\weight(\Ss_k) \geq (1-\eps) \OPT(\Qq_k^{-\delta})$. Then
$$\weight(\Ss_k) \geq (1-\eps) \OPT(\Qq_k^{-\delta}) \geq (1-\eps) \OPT_k(\Qq_k^{-\delta}) \geq (1-\eps) \OPT_k(\Rr).$$
Since there is a value of $k$ for which $\OPT_k(\Rr) = \OPT(\Rr)$, the heaviest family $\Ss_k$ over all $k$ satisfies the conditions of the theorem, because each $\Ss_k$ is an independent set. Hence, it suffices to find the families $\Qq_k$ and $\Ss_k$ for each $k=1,\ldots,n$.

The construction of $\Qq_k$ and $\Ss_k$ depends on whether the squares in $\Rr$ have uniform size or uniform weight. Consider the case when the squares have uniform size. First, apply the kernel of Theorem~\ref{thm:kernel:unit} to $k$, $\delta$, and $\Rr$, and let $\Qq_k \subseteq \Rr$ denote the resulting family. By inspection of the proof of Theorem~\ref{thm:kernel:unit}, $\Qq_k$ is $(\delta/2)$-distant. Now recall that a result of Hunt et al.~\cite[Theorem~5.2]{HuntMRRRS98} shows that \MWISR has a $(1-\eps)$-approximation in time $2^{\Oh_{\lambda}(1/\eps)} \cdot N^{\Oh(1)}$ whenever all the input rectangles are $\lambda$-distant squares of uniform size and non-uniform weight\footnote{Hunt et al.\ state this result explicitly for the unweighted setting~\cite[Theorem~5.2]{HuntMRRRS98}, but the result for the weighted setting follows in exactly the same way. See also~\cite[Section~6.3.1/6.3.5]{proefschrift} for a generalization and a better dependence on $\lambda$.}. We apply this result to the family $\Qq_k^{-\delta}$, which is also $(\delta/2)$-distant, and $\eps$. Let $\Ss_k \subseteq \Qq_k$ be the subfamily of $\Qq_k$ such that $\Ss_k^{-\delta}$ is the output of the approximation algorithm. Note that $\Qq_k$ and $\Ss_k$ are as required, because $\OPT_k(\Qq_k^{-\delta}) \geq \OPT_k(\Rr)$ by Theorem~\ref{thm:kernel:unit}, $\Ss_k^{-\delta}$ is independent by construction, and $\weight(\Ss_k) \geq (1-\eps) \OPT(\Qq_k^{-\delta})$ by construction.

Consider now the case when the input squares have uniform weight. First, apply the kernel of Theorem~\ref{thm:kernel:arb-sizes} to $k$, $\delta$, and $\Rr$, and let $\Qq_k \subseteq \Rr$ denote the resulting family. By Lemma~\ref{lem:app-ply-arb}, $\Qq_k^{-\delta}$ has ply $\Oh_\delta(1)$. Now recall that a result of Chan~\cite[Theorem~3.1]{chan2004note} (see also~\cite[Section~7.3/7.4]{proefschrift} for a generalization) shows that \MWISR has a $(1-\eps)$-approximation in time $2^{\Oh_{\tau}(1/\eps)} \cdot N^{\Oh(1)}$ whenever all the input rectangles are squares of uniform weight, non-uniform size, and ply $\tau$. We apply this result to the family $\Qq_k^{-\delta}$. Let $\Ss_k \subseteq \Qq_k$ be the subfamily of $\Qq_k$ such that $\Ss_k^{-\delta}$ is the output of the approximation algorithm. Note that $\Qq_k$ and $\Ss_k$ are as required, because $\OPT_k(\Qq_k^{-\delta}) \geq \OPT_k(\Rr)$ by Theorem~\ref{thm:kernel:arb-sizes}, $\Ss_k^{-\delta}$ is independent by construction, and $\weight(\Ss_k) \geq (1-\eps) \OPT(\Qq_k^{-\delta})$ by construction.
\end{proof}

\section{Conclusions}\label{sec:conc}

In this paper, we have initiated the study of the shrinking model for parameterized geometric {\sc{Independent Set}}, by giving FPT algorithms and polynomial kernels for the most basic variants.
Most importantly, we have showcased that the shrinking model leads to robust tractability of problems that without this relaxation are hard from the parameterized perspective.
We hope that this is the
start of an interesting direction in the research on combinatorial optimization for geometric problems, as our work raises several concrete open problems.
Can our FPT algorithm and EPTAS for axis-parallel rectangles (Theorems~\ref{thm:main-approx} and~\ref{thm:main-fpt}) be generalized to arbitrary convex polygons, as is the case for the PTAS~\cite{wiese2016independent}?
Is it possible to improve the running time of said FPT algorithm to subexponential, that is, $2^{o(k)}\cdot (nN)^{\Oh(1)}$ for every fixed~$\delta$?
What about polynomial kernels, that is, kernelization procedures with polynomial output guarantees, for more complex objects than squares?
Here, it seems that our arguments apply mutatis mutandis to e.g.\ (unit) disks instead of (unit) squares, but it is conceivable that even the general setting of convex or star-convex polygons can be treated, for an appropriate
definition of shrinking. Also, is there a polynomial kernel for squares of non-uniform size and non-uniform weight? This is not addressed in its full generality by our kernelization algorithms.
Finally, can one give limits to tractability in the shrinking model, by showing $\mathsf{W}[1]$-hardness or the nonexistence of polynomial kernels?
We hope that the future work will address these, as well as multiple other open problems arising from this work.

\paragraph*{Acknowledgements.} 
The first author thanks D\'aniel Marx for helpful discussions about the setting.

\bibliographystyle{abbrv}

\bibliography{references}

\begin{thebibliography}{10}

\bibitem{AdamaszekChalermsookWiese2015}
A.~Adamaszek, P.~Chalermsook, and A.~Wiese.
\newblock {How to Tame Rectangles: Solving Independent Set and Coloring of
  Rectangles via Shrinking}.
\newblock In {\em APPROX/RANDOM 2015}, volume~40 of {\em LIPIcs}, pages 43--60.
  Schloss Dagstuhl--Leibniz-Zentrum f\"ur Informatik, 2015.

\bibitem{AW2013}
A.~Adamaszek and A.~Wiese.
\newblock Approximation schemes for maximum weight independent set of
  rectangles.
\newblock In {\em FOCS 2013}, pages 400--409. IEEE, 2013.

\bibitem{adamaszek2014qptas}
A.~Adamaszek and A.~Wiese.
\newblock A {QPTAS} for maximum weight independent set of polygons with
  polylogarithmically many vertices.
\newblock In {\em {SODA} 2014}, pages 645--656. SIAM, 2014.

\bibitem{AlberF04}
J.~Alber and J.~Fiala.
\newblock Geometric separation and exact solutions for the parameterized
  independent set problem on disk graphs.
\newblock {\em J. Algorithms}, 52(2):134--151, 2004.

\bibitem{Breu-thesis}
H.~Breu.
\newblock {\em Algorithmic Aspects of Constrained Unit Disk Graphs}.
\newblock PhD thesis, University of British Columbia, 1996.

\bibitem{CC2009}
P.~Chalermsook and J.~Chuzhoy.
\newblock Maximum independent set of rectangles.
\newblock In {\em SODA 2009}, pages 892--901. SIAM, 2009.

\bibitem{chan2003polynomial}
T.~M. Chan.
\newblock Polynomial-time approximation schemes for packing and piercing fat
  objects.
\newblock {\em Journal of Algorithms}, 46(2):178--189, 2003.

\bibitem{chan2004note}
T.~M. Chan.
\newblock A note on maximum independent sets in rectangle intersection graphs.
\newblock {\em Information Processing Letters}, 89(1):19--23, 2004.

\bibitem{ChanHarPeled2012}
T.~M. Chan and S.~Har-Peled.
\newblock Approximation algorithms for maximum independent set of pseudo-disks.
\newblock {\em Discrete \& Comp. Geometry}, 48(2):373--392, 2012.

\bibitem{ChuzhoyEne2016}
J.~Chuzhoy and A.~Ene.
\newblock On approximating maximum independent set of rectangles.
\newblock In {\em FOCS 2016}, pages 820--829. {IEEE}, 2016.

\bibitem{CyganFKLMPPS15}
M.~Cygan, F.~V. Fomin, L.~Kowalik, D.~Lokshtanov, D.~Marx, M.~Pilipczuk,
  M.~Pilipczuk, and S.~Saurabh.
\newblock {\em Parameterized Algorithms}.
\newblock Springer, 2015.

\bibitem{DemaineFHT05}
E.~D. Demaine, F.~V. Fomin, M.~T. Hajiaghayi, and D.~M. Thilikos.
\newblock Subexponential parameterized algorithms on bounded-genus graphs and
  $h$-minor-free graphs.
\newblock {\em J. {ACM}}, 52(6):866--893, 2005.

\bibitem{downey1995fixed}
R.~G. Downey and M.~R. Fellows.
\newblock Fixed-parameter tractability and completeness {II}: On completeness
  for {W[1]}.
\newblock {\em Theoretical Computer Science}, 141(1):109--131, 1995.

\bibitem{EJS2005}
T.~Erlebach, K.~Jansen, and E.~Seidel.
\newblock Polynomial-time approximation schemes for geometric intersection
  graphs.
\newblock {\em SIAM J. on Computing}, 34(6):1302--1323, 2005.

\bibitem{EtscheidKMR15}
M.~Etscheid, S.~Kratsch, M.~Mnich, and H.~R{\"{o}}glin.
\newblock Polynomial kernels for weighted problems.
\newblock In {\em {MFCS} 2015}, volume 9235 of {\em LNCS}, pages 287--298.
  Springer, 2015.

\bibitem{FominLMPPS16}
F.~V. Fomin, D.~Lokshtanov, D.~Marx, M.~Pilipczuk, M.~Pilipczuk, and
  S.~Saurabh.
\newblock Subexponential parameterized algorithms for planar and
  apex-minor-free graphs via low treewidth pattern covering.
\newblock In {\em FOCS 2016}, pages 515--524. {IEEE}, 2016.

\bibitem{Har-Peled2014}
S.~Har{-}Peled.
\newblock Quasi-polynomial time approximation scheme for sparse subsets of
  polygons.
\newblock In {\em SOCG 2014}, pages 120--129, 2014.

\bibitem{HochbaumM85}
D.~S. Hochbaum and W.~Maass.
\newblock Approximation schemes for covering and packing problems in image
  processing and {VLSI}.
\newblock {\em J. {ACM}}, 32(1):130--136, 1985.

\bibitem{HuntMRRRS98}
H.~B. {Hunt III}, M.~V. Marathe, V.~Radhakrishnan, S.~S. Ravi, D.~J.
  Rosenkrantz, and R.~E. Stearns.
\newblock {NC}-approximation schemes for {NP-} and {PSPACE}-hard problems for
  geometric graphs.
\newblock {\em J. Algorithms}, 26(2):238--274, 1998.

\bibitem{KleinM14}
P.~N. Klein and D.~Marx.
\newblock A subexponential parameterized algorithm for subset {TSP} on planar
  graphs.
\newblock In {\em {SODA}}, pages 1812--1830. {SIAM}, 2014.

\bibitem{LingasW05}
A.~Lingas and M.~Wahlen.
\newblock A note on maximum independent set and related problems on box graphs.
\newblock {\em Inf. Process. Lett.}, 93(4):169--171, 2005.

\bibitem{Marx07a}
D.~Marx.
\newblock On the optimality of planar and geometric approximation schemes.
\newblock In {\em {FOCS 2007}}, pages 338--348. {IEEE} Computer Society, 2007.

\bibitem{MarxP15}
D.~Marx and M.~Pilipczuk.
\newblock Optimal parameterized algorithms for planar facility location
  problems using {V}oronoi diagrams.
\newblock In {\em {ESA 2015}}, volume 9294 of {\em LNCS}, pages 865--877.
  Springer, 2015.

\bibitem{MullervLvL2013}
T.~M{\"{u}}ller, E.~J. van Leeuwen, and J.~van Leeuwen.
\newblock Integer representations of convex polygon intersection graphs.
\newblock {\em {SIAM} J. Discrete Math.}, 27(1):205--231, 2013.

\bibitem{PilipczukPSL13}
M.~Pilipczuk, M.~Pilipczuk, P.~Sankowski, and E.~J. van Leeuwen.
\newblock Subexponential-time parameterized algorithm for {S}teiner {T}ree on
  planar graphs.
\newblock In {\em {STACS 2013}}, volume~20 of {\em LIPIcs}, pages 353--364.
  Schloss Dagstuhl - Leibniz-Zentrum f\"ur Informatik, 2013.

\bibitem{PilipczukPSL14}
M.~Pilipczuk, M.~Pilipczuk, P.~Sankowski, and E.~J. van Leeuwen.
\newblock Network sparsification for {S}teiner problems on planar and
  bounded-genus graphs.
\newblock In {\em {FOCS 2014}}, pages 276--285. {IEEE} Computer Society, 2014.

\bibitem{SmithW98}
W.~D. Smith and N.~C. Wormald.
\newblock Geometric separator theorems {\&} applications.
\newblock In {\em {FOCS} 1998}, pages 232--243. {IEEE} Computer Society, 1998.

\bibitem{Leeuwen06}
E.~J. van Leeuwen.
\newblock Better approximation schemes for disk graphs.
\newblock In {\em {SWAT} 2006}, volume 4059 of {\em Lecture Notes in Computer
  Science}, pages 316--327. Springer, 2006.

\bibitem{proefschrift}
E.~J. van Leeuwen.
\newblock {\em Optimization and Approximation on Systems of Geometric Objects}.
\newblock PhD thesis, University of Amsterdam, 2009.

\bibitem{wiese2016independent}
A.~Wiese.
\newblock Independent set of convex polygons: From $n^\epsilon$ to $1+\epsilon$
  via shrinking.
\newblock In {\em LATIN 2016}, pages 700--711. Springer, 2016.

\bibitem{Yannakakis}
M.~Yannakakis.
\newblock The complexity of the partial order dimension problem.
\newblock {\em {SIAM} J. Algebraic Discr. Methods}, 3(3):351--358, 1982.

\bibitem{zuck07}
D.~Zuckerman.
\newblock Linear degree extractors and the inapproximability of max clique and
  chromatic number.
\newblock {\em Theory of Computing}, 3:103--128, 2007.

\end{thebibliography}

\end{document}